\documentclass[english,13pt]{article}
\usepackage{tikz}
\usetikzlibrary{shapes.geometric, arrows, positioning}

\tikzstyle{input} = [circle, fill=blue!30, minimum size=1cm, text centered, font=\footnotesize, draw=black]
\tikzstyle{hidden} = [circle, fill=orange!30, minimum size=1cm, text centered, font=\footnotesize, draw=black]
\tikzstyle{output} = [circle, fill=green!30, minimum size=1cm, text centered, font=\footnotesize, draw=black]
\tikzstyle{conv} = [rectangle, fill=purple!30, minimum size=1cm, text centered, font=\footnotesize, draw=black]
\tikzstyle{pool} = [rectangle, fill=cyan!30, minimum size=1cm, text centered, font=\footnotesize, draw=black]
\tikzstyle{fc} = [rectangle, fill=orange!30, minimum size=1cm, text centered, font=\footnotesize, draw=black]
\tikzstyle{rnn} = [rectangle, fill=orange!30, minimum size=1cm, text centered, font=\footnotesize, draw=black]
\tikzstyle{generator} = [rectangle, fill=purple!30, minimum size=1cm, text centered, font=\footnotesize, draw=black]
\tikzstyle{discriminator} = [rectangle, fill=cyan!30, minimum size=1cm, text centered, font=\footnotesize, draw=black]
\tikzstyle{arrow} = [->, thick, draw=black]

\usepackage[T1]{fontenc}
\usepackage[latin1]{inputenc}
\usepackage{geometry}
\usepackage{float}
\usepackage{mathrsfs}
\usepackage{amsmath}
\usepackage{amssymb}
\usepackage{graphicx}
\usepackage{hyperref}
\usepackage[color=yellow]{todonotes}
\hypersetup{
    colorlinks=true,
    linkcolor=blue,
    filecolor=magenta,      
    urlcolor=cyan,
    citecolor = blue,
}

\usepackage{yhmath}
\makeatletter

\usepackage{bm}
\numberwithin{equation}{section}

\usepackage{algorithm}
 \usepackage{algorithmicx}
\floatstyle{ruled}
\newfloat{algorithm}{tbp}{loa}
\providecommand{\algorithmname}{Algorithm}
\floatname{algorithm}{\protect\algorithmname}

\usepackage{algorithm}
 \usepackage{algorithmicx}

\usepackage{amsthm}

\usepackage{mathrsfs}
\usepackage{amssymb}

\usepackage{amsfonts}

\usepackage{epsfig}

\usepackage{bm}

\usepackage{mathrsfs}

\usepackage{enumerate}

\@ifundefined{definecolor}{\@ifundefined{definecolor}
 {\@ifundefined{definecolor}
 {\usepackage{color}}{}
}{}
}{}

\usepackage{subfig}\usepackage[all]{xy}

\newtheorem{theorem}{Theorem}[section]
\newtheorem{lem}{Lemma}[section]

\newcounter{hypA}
\newenvironment{hypA}{\refstepcounter{hypA}\begin{itemize}
  \item[({\bf A\arabic{hypA}})]}{\end{itemize}}
\newcounter{hypB}

\newcounter{hypD}
\newenvironment{hypD}{\refstepcounter{hypD}\begin{itemize}
 \item[({\bf D\arabic{hypD}})]}{\end{itemize}}

\usepackage{babel}\date{}

\makeatother

\begin{document}

\begin{center}

{\Large \textbf{Parameter Estimation for Partially Observed McKean-Vlasov Diffusions}}

\vspace{0.5cm}

AJAY JASRA$^{1}$, MOHAMED MAAMA$^{2}$ \& RAUL TEMPONE$^{2,3}$

{\footnotesize $^{1}$School of Data Science,  The Chinese University of Hong Kong,  Shenzhen,  Shenzhen,  CN.}\\
{\footnotesize $^{2}$Applied Mathematics and Computational Science Program, \\ Computer, Electrical and Mathematical Sciences and Engineering Division, \\ King Abdullah University of Science and Technology, Thuwal, 23955-6900, KSA.} \\
{\footnotesize $^{3}$Chair of Mathematics for Uncertainty Quantification, 
RWTH Aachen University, 52062 Aachen, Germany.}\\
{\footnotesize E-Mail:\,} \texttt{\emph{\footnotesize ajayjasra@cuhk.edu.cn}}, 
\texttt{\emph{\footnotesize maama.mohamed@gmail.com}},  \\ \texttt{\emph{\footnotesize raul.tempone@kaust.edu.sa}}

\begin{abstract}
In this article we consider likelihood-based estimation of static parameters for a class of partially observed McKean-Vlasov (POMV) diffusion process with discrete-time observations over a fixed time interval.  In particular,  using the framework
of \cite{ub_grad_new} we develop a new randomized multilevel Monte Carlo method for estimating
the parameters,  based upon Markovian stochastic approximation methodology.   New Markov chain Monte Carlo
algorithms  for the POMV model  are introduced facilitating the application of \cite{ub_grad_new}.  We prove, under assumptions,  that the 
expectation of our estimator is biased,  but with expected small and controllable bias.  Our approach is implemented on several examples.
\end{abstract}

\end{center}

\section{Introduction}

We consider the problem of parameter estimation for partially observed McKean-Vlasov (MV) stochastic differential equations (SDE). MV SDEs find a wide variety of applications in mathematics,  neuroscience and beyond; see for instance \cite{bald,crisan,fin,biol} and the references therein. In our context,  we assume that data are observed in discrete time and are conditionally independent given the position of the MV SDE at the observation time.  Given an overall fixed time window we seek to estimate the parameters of the overall model.

The problem that we are investigating is rather classical in the context of state-space models \cite{bain, cappe}.
In most of the literature that we are aware of,  concerning partially observed continuous-time processes,  much of the work regarding parameter estimation can be found in the context of regular diffusions or jump-diffusion processes;  see \cite{ub_grad_new,beskos,beskos1,chada_ub,graham,ub_grad,jasra_bpe_sde} for a non-exhaustive list of articles.  In most of the afore-mentioned articles,  in order to estimate the parameters of the model
one considers either a likelihood-based or Bayesian approach.  In either context one has to resort to numerical methods to estimate the parameters of the model,  often using Markov chain Monte Carlo (MCMC) in some guise or another. This is because the likelihood function for the model of interest is not available analytically; see for instance \cite{cappe} for a full review.
In the case of parameter estimation associated to partially observed MV SDEs we consider a likelihood-based approach.  In the case where the MV SDE has been fully observed there are several works (see \cite{nik,wen}
and the references therein) on parameter estimation,  but there seems few papers that have focussed on the partially observed case.  The challenges associated to this scenario, relative to the ordinary diffusion case,  is the fact that the coefficients of the MV SDE are typically unavailable,  due to the dependence on the typically unknown law of the MV SDE.  From a numerical perspective this introduces several additional challenges and this is what we intend to deal with in this article.

In order to describe the solution that is proposed,  analyzed and implemented in this article,  we begin by going into a little more detail on the model.  The main point which is apparent is that,  on deriving an appropriate equation for the likelihood function it becomes clear that it cannot be computed,  let alone a gradient, which can be leveraged for maximizing the likelihood.  In the ordinary diffusion case,  this can be simplified by simply employing time-discretization methods of the likelihood itself e.g.~\cite{beskos,ub_grad} and then either again time discretizing the diffusion (e.g.~by Euler-Maruyama) or using exact simulation methods (e.g.~\cite{beskos}).  As the latter seems rather challenging,  we focus on using the Euler-Maruyama method,  however, this itself does not lend itself to the MCMC approaches (for example) that have been used in the afore-mentioned references e.g.~\cite{graham,ub_grad}.  This is because even under time-discretization,  exact simulation of the resulting stochastic process in intractable and being able to do this is a core-property of most of the techniques that are often used in the literature.

In this paper we construct a new gradient-based method for estimating the static (i.e.~non time varying) parameters of a partially observed MV SDE model.  Our main starting point is the approach of \cite{ub_grad_new}, which was introduced in the case where the SDE is standard (elliptic SDE) and relies upon Markovian stochastic approximation (e.g.~\cite{andr}) and combining it with unbiased and multilevel Monte Carlo \cite{giles,giles1,hein,rhee} approaches.  It is particularly fast relative to some competing methods (e.g.~\cite{ub_grad}) whilst providing parameter estimates that are unbiased,  in that the estimators suffer from no time discretization bias, despite the fact that the approach only resorts to the application of time discretizations.  As the method is so-called embarrassingly parallel it can have optimal Monte Carlo convergence rates (see e.g.~the discussion in \cite{ub_pf}).

Our contributions are as follows:
\begin{itemize}
\item{We develop a new version of the approach of \cite{ub_grad_new} with a non-trivial extension to the partially observed MV SDE.}
\item{We show that the estimator,  under assumptions, is `close' in expectation to the true maximum likelihood estimator (MLE).}
\item{We implement our methodology,  demonstrating the efficacy of our approach.}
\end{itemize}
To expand further on the above bullets, in the context of the first,  based upon the ideas in \cite{po_mv} we develop new MCMC methods that can be used to implement the framework in \cite{ub_grad_new}.  In terms of the second,  for partially observed MV SDEs we terminate the maximal amount of time discretization that could be used, i.e.~how precise the Euler-Maruyama method is,  leading to a type of randomized MLMC algorithm. This in turn means that although we can show that our method is,  in expectation,  equal to a particular unique value,  that value is only likely to be `close' to the true MLE (i.e.~without time discretization).  

This article is structured as follows.  In Section \ref{sec:problem} we detail the estimation problem that is to be considered.  In Section \ref{sec:meth_theory} we present the methodology to be used as well as our theoretical results.  In Section \ref{sec:numerics} features our numerical results. The proofs associated to our theoretical results can be found in Appendix \ref{app:theory}.

\section{Problem Formulation}\label{sec:problem}

\subsection{Model}

We consider the stochastic differential equation (SDE) with $X_0=x_0\in\mathbb{R}^d$, $\theta\in\Theta$ fixed:
\begin{equation}\label{eq:sde}
dX_t = a_{\theta}\left(X_t,\overline{\xi}_{\theta}(X_t,\mu_{t,\theta})\right)dt + \sigma\left(X_t\right)dW_t
\end{equation}
where 
$$
\overline{\xi}_{\theta}(X_t,\mu_{t,\theta})  =  \int_{\mathbb{R}^d}\xi_{\theta}(X_t,x)\mu_{t,\theta}(dx)
$$
$\{W_t\}_{t\geq 0}$ is a standard $d-$dimensional Brownian motion,  for each $\theta\in\Theta\subset\mathbb{R}^{d_{\theta}}$,  $\xi_{\theta}:\mathbb{R}^{2d}\rightarrow\mathbb{R}$, $a_{\theta}:\mathbb{R}^d\times\mathbb{R}\rightarrow\mathbb{R}^d$, $\sigma:\mathbb{R}^d\rightarrow\mathbb{R}^{d\times d}$, $\mu_{t,\theta}$ is the law of $X_t$ and $\mu_{0,\theta}(dx)=\delta_{\{x_0\}}(dx)$ (Dirac measure on the set $\{x_0\}$ the starting point is taken as fixed).

We make the following assumption throughout the paper without further mention; it is known \cite{ver} that this ensures the existence of a strong solution.  Set $\mathcal{P}(\mathbb{R}^d)$  the probability measures on the measurable space $(\mathbb{R}^d,\mathcal{B}(\mathbb{R}^d))$ with $\mathcal{B}(\mathbb{R}^d)$ the Borel sets on $\mathbb{R}^d$. We write $\mathcal{C}_b^k(\mathbb{R}^{d_1},\mathbb{R}^{d_2})$ as the collection of $k-$times continuously differentiable functions from $\mathbb{R}^{d_1}$ to $\mathbb{R}^{d_2}$ with bounded derivatives of order 1 up-to $k$.  

\begin{hypD}
\begin{enumerate}
\item{For each $(\theta,\mu)\in\Theta\times\mathcal{P}(\mathbb{R}^d)$, $(a_{\theta}\left(\cdot,\overline{\xi}_{\theta}(\cdot,\mu)\right)
,\xi_{\theta}(\cdot))
\in\mathcal{C}_b^2(\mathbb{R}^{d+1},\mathbb{R}^d)\times\mathcal{C}_b^2(\mathbb{R}^{2d},\mathbb{R})$.}
\item{$\sigma(\cdot)\in\mathcal{C}_b^2(\mathbb{R}^{d},\mathbb{R}^{d\times d})$.}
\item{Set $\Sigma(x)=\sigma(x)\sigma(x)^{\top}$, then for any $x\in\mathbb{R}^d$, $\Sigma(x)$ is positive definite.}
\end{enumerate}
\end{hypD}

We denote by $P_{\mu_{t-1,\theta},t,\theta}(x_{t-1},dx_t)$ the conditional law of $X_t$ (as given in \eqref{eq:sde}) given $\mathscr{F}_{t-1}$ (the natual filtration of the process), for $t\geq 1$; that is, the transition kernel over unit time.  We consider a discrete time observation
process $Y_1,Y_2,\dots$,  $Y_t\in\mathsf{Y}$, that are assumed, for notational convenience, to be observed at unit times. Conditional on the position $X_t$, $t\in\mathbb{N}$ of \eqref{eq:sde}, the random variable $Y_t$ is assumed to be independent of all other random variables, with a bounded and positive probability density $G_{\theta}(x_t,y_t)$.

Let $\varphi:\mathbb{R}^d\rightarrow\mathbb{R}$ be a bounded and measurable function then we define the filtering expectation for $t\in\mathbb{N}$ as:
$$
\pi_{t,\theta}(\varphi) := \frac{\int_{\mathbb{R}^{dt}}\varphi(x_t)\left\{\prod_{p=1}^t G_{\theta}(x_p,y_p)\right\}\prod_{p=1}^tP_{\mu_{p-1,\theta},p,\theta}(x_{p-1},dx_p)}{\int_{\mathbb{R}^{dt}}\left\{\prod_{p=1}^t G_{\theta}(x_p,y_p)\right\}\prod_{p=1}^tP_{\mu_{p-1,\theta},p,\theta}(x_{p-1},dx_p)}.
$$
We remark that one does not need $\varphi$
and each of the $G_{\theta}(\cdot,y_t)$ to be bounded, but it will simplify the resulting exposition to do so.
For $T\in\mathbb{N}$ fixed, our objective is to maximize the log-likelihood w.r.t.~$\theta$:
$$
\log\left(p_{\theta}(y_1,\dots,y_T)\right) = \log\left(\int_{\mathbb{R}^{dT}}
\left\{\prod_{p=1}^T G_{\theta}(x_p,y_p)\right\}\prod_{p=1}^T P_{\mu_{p-1,\theta},p,\theta}(x_{p-1},dx_p)\right).
$$
We use $\mathbb{E}_{\theta}[\cdot]$ to denote expectations w.r.t.~the law diffusion process,  ($\mathbb{P}_{\theta}$),  so we can write
$$
p_{\theta}(y_1,\dots,y_T) = \mathbb{E}_{\theta}\left[\left\{\prod_{p=1}^T G_{\theta}(X_p,y_p)\right\}\right].
$$

By using the Cameron-Martin theorem:
$$
\mathbb{E}_{\theta}\left[\left\{\prod_{p=1}^T G_{\theta}(X_p,y_p)\right\}\right] = \mathbb{E}_{\mathbb{Q}}\left[
\left\{\prod_{p=1}^T G_{\theta}(X_p,y_p)\right\}
\frac{d\mathbb{P}_{\theta}}{d\mathbb{Q}}\right]
$$
where 
$$
\frac{d\mathbb{P}_{\theta}}{d\mathbb{Q}} = \exp\Big\{-\frac{1}{2}\int_{0}^T \|b_{\theta}\left(X_s,\overline{\xi}_{\theta}(X_s,\mu_{s,\theta})\right)\|^2ds + \int_{0}^{T}
b_{\theta}\left(X_s,\overline{\xi}_{\theta}(X_s,\mu_{s,\theta})\right)^{\top}d\widetilde{W}_s\Big\}
$$
$\mathbb{Q}$ is a probability  on the path $\{X_t\}_{t\geq 0}$,
$$
b_{\theta}\left(X_s,\overline{\xi}_{\theta}(X_s,\mu_{s,\theta})\right) = \Sigma(X_s)^{-1}\sigma(X_s)^{\top}
a_{\theta}\left(X_s,\overline{\xi}_{\theta}(X_s,\mu_{s,\theta})\right)
$$ 
is a $d-$vector and under $\mathbb{Q}$, $\{X_t\}_{t>0}$ solves $dX_t = \sigma(X_t)d\widetilde{W}_t$ where 
$\{\widetilde{W}_t\}_{t\geq 0}$ is a standard $d-$dimensional Brownian motion under $\mathbb{Q}$.  It is assumed that $a_{\theta},G_{\theta}$ and $\sigma$ are such that $\left\{\prod_{p=1}^T G_{\theta}(x_p,y_p)\right\}
\frac{d\mathbb{P}_{\theta}}{d\mathbb{Q}}$ is $\mathbb{Q}-$integrable for each fixed $T\geq 0$.
As it will be useful below, we write
$$
\frac{d\mathbb{P}_{\theta}}{d\mathbb{Q}} = \exp\Big\{-\frac{1}{2}\int_{0}^T \|b_{\theta}\left(X_s,\overline{\xi}_{\theta}(X_s,\mu_{s,\theta})\right)\|^2ds + \int_{0}^{T}b_{\theta}\left(X_s,\overline{\xi}_{\theta}(X_s,\mu_{s,\theta})\right)^{\top}
\Sigma(X_s)^{-1}\sigma(X_s)^{\top} dX_s\Big\}.
$$
Now we shall explicitly assume: 
\begin{quote}
$G_\theta$ and $\frac{d\mathbb{P}_{\theta}}{d\mathbb{Q}}$ are differentiable w.r.t.~$\theta$.
\end{quote}
Then one has, under very minor regularity conditions that
\begin{equation}\label{eq:grad_like}
\nabla_{\theta}\log\left\{
\mathbb{E}_{\theta}\left[\left\{\prod_{p=1}^T G_{\theta}(X_p,y_p)\right\}\right]
\right\} = \mathbb{E}_{\overline{\mathbb{P}}_{\theta}}\left[\nabla_{\theta}\log\left(
\left\{\prod_{p=1}^T G_{\theta}(X_p,y_p)\right\}
\frac{d\mathbb{P}_{\theta}}{d\mathbb{Q}}
\right)\right]
\end{equation}
where 
$$
\overline{\mathbb{P}}_{\theta}\left(d\{x_s\}_{s\in[0,T]}\right)= \frac{\left\{\prod_{p=1}^T G_{\theta}(x_p,y_p)\right\}\mathbb{P}_{\theta}\left(d\{x_s\}_{s\in[0,T]}\right)}{p_{\theta}(y_1,\dots,y_T)}.
$$ 
We denote by $\theta_{\star}^L\rightarrow\theta^{\star}$ is the (assumed) unique maximizer of the continuous-time likelihood.

This is a construction developed in \cite{ub_grad} for regular SDEs and we shall show how 
the methodology in \cite{ub_grad_new,po_mv} can be extended to this case.   
We remark that in the expression
$$
\nabla_{\theta}\log\left(
\left\{\prod_{p=1}^T G_{\theta}(X_p,y_p)\right\}
\frac{d\mathbb{P}_{\theta}}{d\mathbb{Q}}
\right)
$$
one must compute derivatives of $a_{\theta}$ along the way and this in turn depends on the gradient of the law $\mu_{t,\theta}$.  In practice we will use finite difference approximations and this is discussed later on.

\subsection{Time Discretization}

In practice,  one cannot work with the continuous-time formulation above,  so we shall consider a standard time-discretization and an associated optimization problem.  This in turn will be approximated by using numerical methods,  which will be the topic of Section \ref{sec:meth_theory}.

Let $l\in\mathbb{N}_0=\mathbb{N}\cup\{0\}$ be given and set $\Delta_l=2^{-l}$.  We consider the first order Euler-Maruyama time discretization of \eqref{eq:sde} for $k\in\{0,1\dots,\Delta_l^{-1}T-1\}$,  $\widetilde{x}_0=x_0$ (the starting position of \eqref{eq:sde} is a known point $x_0$)
\begin{equation}\label{eq:sde_disc}
\widetilde{X}_{(k+1)\Delta_l} = \widetilde{X}_{k\Delta_l} + a_{\theta}\left(\widetilde{X}_{k\Delta_l},\overline{\xi}_{\theta}(\widetilde{X}_{k\Delta_l},\mu_{k\Delta_l,\theta}^l)\right)dt + \sigma\left(\widetilde{X}_{k\Delta_l}\right)[W_{(k+1)\Delta_l}-W_{k\Delta_l}]
\end{equation}
where $\mu_{k\Delta_l,\theta}^l$ denotes the law of the time discretized process at time $k\Delta_l$. 
Recall $P_{\mu_{t-1,\theta},t,\theta}(x_{t-1},dx_t)$ ($t\in\{1,\dots,T\}$) from the previous section: we denote by
$P_{\mu_{t-1,\theta}^l,t,\theta}^l(x_{t-1},dx_t)$
the conditional law of $\widetilde{X}_t$ (as given in \eqref{eq:sde_disc}); the time discretized transition kernel over unit time. 

We can now define a time discretized version of the log-likelihood (c.f.~\eqref{eq:grad_like}) and its gradient in the following manner.  We have
$$
\log(p_{\theta}^l(y_1,\dots,y_T)) := \log\left(
\int_{\mathbb{R}^{dT}} \left\{\prod_{k=1}^TG_{\theta}(x_k,y_k) \right\} \prod_{k=1}^T
P_{\mu_{k-1,\theta}^l,k,\theta}^l(x_{k-1},dx_k)
\right).
$$
Set
\begin{eqnarray}
H_l(\theta,\widetilde{x}_0,\widetilde{x}_{\Delta_l},\dots,\widetilde{x}_T) & = & \sum_{k=1}^T \nabla_{\theta}\log\left\{G_{\theta}(\widetilde{x}_k,y_k)\right\} +
\sum_{k=0}^{\Delta_l^{-1}T-1}\Bigg\{-\frac{\Delta_l}{2}\nabla_{\theta}\|b_{\theta}\left(\widetilde{x}_{k\Delta_l},\overline{\xi}_{\theta}(\widetilde{x}_{k\Delta_l},\mu_{k\Delta_l,\theta}^l)\right)\|^2 + \nonumber\\ & & \nabla_{\theta}
\{b_{\theta}\left(\widetilde{x}_{k\Delta_l},\overline{\xi}_{\theta}(\widetilde{x}_{k\Delta_l},\mu_{k\Delta_l,\theta}^l)\right)^{\top}
\Sigma(\widetilde{x}_{k\Delta_l})^{-1}\sigma(\widetilde{x}_{k\Delta_l})^{\top}[\widetilde{x}_{(k+1)\Delta_l}-\widetilde{x}_{k\Delta_l}]\}
\Bigg\}\label{eq:Hl_def}
\end{eqnarray}
Then the gradient of  time discretized version of the log-likelihood denoted $\nabla_{\theta}\log(p_{\theta}^l(y_1,\dots,y_T))$ is equal to
\begin{equation}\label{eq:ll_grad_disc}
\nabla_{\theta}\log(p_{\theta}^l(y_1,\dots,y_T)) = \frac{\mathbb{E}_{\theta}
\left[\left\{\prod_{k=1}^TG_{\theta}(\widetilde{X}_k,y_k) \right\}
H_l(\theta,\widetilde{X}_0,\widetilde{X}_{\Delta_l},\dots,\widetilde{X}_T)
\right]
}{\mathbb{E}_{\theta}
\left[\left\{\prod_{k=1}^TG_{\theta}(\widetilde{X}_k,y_k) \right\}\right]}.
\end{equation}
The objective is then to maximize $\log(p_{\theta}^l(y_1,\dots,y_T))$ by using the gradient defined above.  We shall see that we will be able to do this task up-to some further approximation which shall be explained in Section 
\ref{sec:meth_theory}.

\section{Methodology and Theory}\label{sec:meth_theory}

\subsection{Introduction}

We now describe the methodology that we will use to obtain parameter estimates of $\theta$.  We will introduce a non-trivial extension of the approach in \cite{ub_grad_new}.  The main complication that we need to address is the fact that numerical simulation of \eqref{eq:sde_disc} is not feasible without the approximating the laws $\mu_{k\Delta_l,\theta}^l$.  The resulting algorithms, as used in \cite{ub_grad_new},  then need to be adapted and appropriately analyzed.
To that end we begin in Section \ref{sec:law_approx} by describing the well-known method of \cite{basic_method} for precisely approximating $\mu_{k\Delta_l,\theta}^l$.  Then as \eqref{eq:ll_grad_disc} can be thought of a type of expectation w.r.t.~a smoothing distribution,  we show how such an expectation can be approximated using MCMC.  The details are given in Section \ref{sec:smooth_approx} where the notion of a smoothing distribution is clarified. 
We then show how the differences of expectations of smoothers,  associated to two different levels of time discretization can be approximated by coupled MCMC (Section \ref{sec:smooth_pair_approx}); such a concept is needed in the sequel. We then describe the strategy for optimization based on stochastic approximation (Section \ref{sec:sa}) and our final algorithm (Section \ref{sec:fa}).  The algorithm is discussed in Section \ref{sec:disc_meth}, where several of the algorithmic choices are analyzed.
Our theoretical results are given in Section \ref{sec:theory}.  From herein we remove the $\widetilde{\cdot}$ notation as we are now working only with a time discretized process.  Throughout the section the level $l$ of time discretization will be fixed unless otherwise stated.

\subsection{Approximating the Laws}\label{sec:law_approx}

We now detail the method of \cite{basic_method} for allowing one to approximate the law
$\mu_{t,\theta}^l$ for each $t\in\{1,\dots,T\}$.  To that end we will need the notation that
$\mathcal{N}_d(c,\Sigma)$ denotes the $d-$dimensional Gaussian distribution with mean $c$ and covariance matrix $\Sigma$. $I_d$ is the $d\times d$
identity matrix and $\stackrel{\textrm{ind}}{\sim}$ denotes independently distributed as.
The approach is given in Algorithm \ref{alg:basic_method} from \cite{basic_method}.   
Note that in Step 2.~when $k=1$ the points $X_{t-1}^1,\dots,X_{t-1}^N$ are available via the empirical
measure that has to be specified in Step 1..
We remark that
Algorithm \ref{alg:basic_method} can be used to approximate expectations w.r.t.~$\mu_{t,\theta}^l$
and indeed on the grid in-between time $t-1$ and $t$.    

\begin{algorithm}[h]
\begin{enumerate}
\item{Input $l\in\mathbb{N}_0$ the level of discretization, $N\in\mathbb{N}$ the number of particles, $\theta\in\Theta$,  $t\in\{1,\dots,T\}$. If $t=1$ set $\mu_{0,\theta}^{l,N}(dx)=\delta_{\{x_0\}}(dx)$ otherwise input an empirical measure $\mu_{t-1,\theta}^{l,N}(dx)=\tfrac{1}{N}\sum_{i=1}^N\delta_{\{X_{t-1}^i\}}(dx)$. Set $k=1$.}
\item{For $i\in\{1,\dots,N\}$ generate:
\begin{align*}
X_{t-1+k\Delta_l}^i & =  X_{t-1+(k-1)\Delta_l}^i + a_{\theta}\left(X_{t-1+(k-1)\Delta_l}^i,\overline{\xi}_{\theta}(X_{t-1+(k-1)\Delta_l}^i,\mu_{t-1+(k-1)\Delta_l,\theta}^{l,N})\right) + \\ & \sigma\left(X_{t-1+(k-1)\Delta_l}^i\right)\left[W_{t-1+k\Delta_l}^i - W_{t-1+(k-1)\Delta_l}^i\right]
\end{align*}
where
\begin{eqnarray*}
\overline{\xi}_{\theta}(X_{t-1+(k-1)\Delta_l}^i,\mu_{t-1+(k-1)\Delta_l,\theta}^{l,N}) & = & \frac{1}{N}\sum_{j=1}^N \xi_{\theta}(X_{t-1+(k-1)\Delta_l}^i,X_{t-1+(k-1)\Delta_l}^j)\\
\mu_{t-1+(k-1)\Delta_l,\theta}^{l,N}(dx) & = & \frac{1}{N}\sum_{j=1}^N\delta_{\{X_{t-1+(k-1)\Delta_l}^j\}}(dx) \\
\left[W_{t-1+k\Delta_l}^i - W_{t-1+(k-1)\Delta_l}^i\right] & \stackrel{\textrm{ind}}{\sim} & \mathcal{N}_{d}(0,\Delta_l I_d).
\end{eqnarray*}
Set $k=k+1$, if $k=\Delta_l^{-1}+1$ go to step 3.~otherwise go to the start of step 2..}
\item{Output all the required laws $\mu_{t-1+\Delta_l,\theta}^N,\dots,\mu_{t,\theta}^N$.}
\end{enumerate}
\caption{Approximating the Laws when starting with a particle approximation at time $t-1$, $t\in\{1,\dots,T\}$.}
\label{alg:basic_method}
\end{algorithm}

\subsection{Approximating the Smoother}\label{sec:smooth_approx}

\subsubsection{Modified Objective Function}\label{sec:mod_obj}

We begin with some explanation of what we mean by the smoother in this context.  We write
$Q_{\mu,\theta}^l(x,dy)$ as the Gaussian Markov kernel on $(\mathbb{R}^d,\mathcal{B}(\mathbb{R}^d))$ associated to a $\Delta_l$ time step
of \eqref{eq:sde_disc},  with input measure $\mu\in\mathcal{P}(\mathbb{R}^d)$. That is,  for any $t\in\{1,\dots,T\}$, we have
$$
P_{\mu_{t-1,\theta}^l,t,\theta}^l(x_{t-1},dx_t) = \int_{\mathbb{R}^{d(\Delta_l^{-1}-1)}}
\prod_{k=1}^{\Delta_{l}^{-1}} Q_{\mu_{t-1+(k-1)\Delta_l,\theta}^l,\theta}^l(x_{t-1+(k-1)\Delta_l},dx_{t-1+k\Delta_l}).
$$
Then we can write the gradient of the log-likelihood in \eqref{eq:ll_grad_disc} as an expectation w.r.t.~a probability measure which we will call the smoother.  Set $u_t=(x_{t-1+\Delta_l},\dots,x_t)$,  
and write
$$
\overline{P}_{\mu_{t-1,\theta}^l,t,\theta}^l(x_{t-1},du_t) := \prod_{k=1}^{\Delta_{l}^{-1}} Q_{\mu_{t-1+(k-1)\Delta_l,\theta}^l,\theta}^l(x_{t-1+(k-1)\Delta_l},dx_{t-1+k\Delta_l})
$$
then the smoother is
$$
\overline{\pi}^l_{\theta}\left(d(u_1,\dots,u_T)\right) := \frac{\left\{\prod_{k=1}^T G_{\theta}(x_k,y_k)\right\}
\prod_{k=1}^T \overline{P}_{\mu_{t-1,\theta}^l,t,\theta}^l(x_{t-1},du_t)
}{\int_{\mathsf{E}_l^T}
\left\{\prod_{k=1}^T G_{\theta}(x_k,y_k)\right\}
\prod_{k=1}^T \overline{P}_{\mu_{t-1,\theta}^l,t,\theta}^l(x_{t-1},du_t)
}
$$
where $\mathsf{E}_l=\mathbb{R}^{d\Delta_l^{-1}}$.  It then follows that
$$
\nabla_{\theta}\log(p_{\theta}^l(y_1,\dots,y_T)) = \int_{\mathsf{E}_l^T}
H_l(\theta,u_1,\dots,u_T)
\overline{\pi}_{\theta}^l\left(d(u_1,\dots,u_T)\right)
$$
where we have surpressed $x_0$ from the notation for $H_l$.

Ideally,  our objective would now to be to sample from $\overline{\pi}^l$ and then one can compute a Monte Carlo
approximation of $\nabla_{\theta}\log(p_{\theta}^l(y_1,\dots,y_T))$.  The issue here is that, of course,  even if the the laws of the SDE were known,  one has to resort to MCMC,  however, we have the further complication that we
cannot use MCMC to sample from $\overline{\pi}^l$.  In what follows we shall modify the problem somewhat,  in that we will consider a modified gradient expression which we will now describe,  in such a way that we are able to use MCMC to approximate expectations w.r.t.~a modified smoother and that the `bias' of the smoother can be controlled by a simulation parameter.  

Consider sequentially sampling Algorithm \ref{alg:basic_method} from time
$1$ to time $T$ using the approximated laws from the previous time step,  to initialize the next time step; throughout $N$ is fixed and so is $\theta$.
At time $t$, Write all the simulated random variables $(X_{t-1+\Delta_l}^1,\dots,X_{t-1+\Delta_l}^N),\dots,
(X_{t}^1,\dots,X_{t}^N)$ from Algorithm \ref{alg:basic_method} as $\overline{u}_t=u_t^{1:N}\in\mathsf{E}_l^N$.
Finally write the joint probability of $(\overline{u}_1,\dots,\overline{u}_T)$ as $\overline{\mathbb{P}}_{\theta}^N$
and associated expectation $\overline{\mathbb{E}}_{\theta}^N$. 
We will write for $t\in\{1,\dots,T\}$
$$
\overline{P}_{\mu_{t-1,\theta}^{l,N},t,\theta}^l(x_{t-1},du_t) =  \prod_{k=1}^{\Delta_{l}^{-1}} Q_{\mu_{t-1+(k-1)\Delta_l,\theta}^{l,N},\theta}^l(x_{t-1+(k-1)\Delta_l},dx_{t-1+k\Delta_l}).
$$
We will now consider an MCMC method to sample from
$$
\overline{\pi}^{l,N}_{\theta}\left(d(u_1,\dots,u_T)\right) = \frac{
\left\{\prod_{k=1}^T G_{\theta}(x_k,y_k)\right\}
\overline{\mathbb{E}}_{\theta}^N\left[\prod_{t=1}^T \overline{P}_{\mu_{t-1,\theta}^{l,N},t,\theta}^l(x_{t-1},du_t)
\right]
}{
\int_{\mathsf{E}_l^T}
\left\{\prod_{k=1}^T G_{\theta}(x_k,y_k)\right\}
\overline{\mathbb{E}}_{\theta}^N\left[
\prod_{t=1}^T\overline{P}_{\mu_{t-1,\theta}^{l,N},t,\theta}^l(x_{t-1},du_t)
\right]
}.
$$
Subequently,  we will focus on approximation the gradient of
$$
\log(p_{\theta}^{l,N}(y_1,\dots,y_T)) = \int_{\mathsf{E}_l^T}
\left\{\prod_{k=1}^T G_{\theta}(x_k,y_k)\right\}
\overline{\mathbb{E}}_{\theta}^N\left[
\prod_{t=1}^T \overline{P}_{\mu_{t-1,\theta}^{l,N},t,\theta}^l(x_{t-1},du_t)
\right]
$$
via the expression
\begin{equation}\label{eq:mod_grad_ll}
\nabla_{\theta}\log(p_{\theta}^{l,N}(y_1,\dots,y_T)) = \int_{\mathsf{E}_l^T}
\overline{\mathbb{E}}_{\theta}^N\left[H_l^N(\theta,s_1,\dots,s_T)\right]
\overline{\pi}^{l,N}_{\theta}\left(d(u_1,\dots,u_T)\right)
\end{equation}
where $s_t=(u_t,\overline{u}_t)$,  $t\in\{1,\dots,T\}$ and $H_l^N$ is as \eqref{eq:Hl_def}, except we have plugged in the empirical measures that have been produced via Algorithm \ref{alg:basic_method} .    We 
note that in $H_l^N$ 
we will need some further approximation 
constructed via $\overline{\mathbb{P}}_{\theta}^N\left(d(\overline{u}_1,\dots,\overline{u}_T)\right)$
to compute it
and hence \eqref{eq:mod_grad_ll} \emph{will be further modified} and is explained in Section \ref{sec:hl_comp} below.
For now we focus on sampling from $\overline{\pi}^{l,N}_{\theta}$ and suppose that $H_l^N$ can be computed.

\subsubsection{MCMC Method}

We present our MCMC method to approximate expectations w.r.t.~$\overline{\pi}^{l,N}_{\theta}$ in Algorithm 
\ref{alg:cond_pf_0}.  The approach detailed in Algorithm  \ref{alg:cond_pf_0} describes how to simulate a Markov kernel 
$K_{l,\theta}:\mathsf{E}_l^T\rightarrow\mathcal{P}(\mathsf{E}_l^T)$ and is a new type of conditional particle filter that was originally developed in \cite{andrieu}.  
The main innovation is Step 2.~of Algorithm  \ref{alg:cond_pf_0} which inputs a sequential version of Algorithm \ref{alg:basic_method} to approximate the laws in the kernels 
$Q_{\mu_{(k-1)\Delta_l,\theta}^{l},\theta}^l$ and is indeed simply a conditional particle filter that was used in
\cite{po_mv}.
We remark that in practice it is likely that one simulates the approximation of the measures (Step 2.~of Algorithm  \ref{alg:cond_pf_0}) sequentially in each time step of the main algorithm (i.e.~Step 3.~of Algorithm  \ref{alg:cond_pf_0}),  but we wanted to make it clear that we are simulating from 
$\overline{\mathbb{P}}_{\theta}^N$
independently of all other random variables that are generated in Algorithm  \ref{alg:cond_pf_0}.

Using the approach that was developed in \cite{andrieu},  one can verify that under minimal conditions that the invariant measure of the kernel $K_{l,\theta}$ is exactly $\overline{\pi}^{l,N}_{\theta}$.  As a result one can approximate \eqref{eq:mod_grad_ll} by iteratively simulating $K_{l,\theta}$ and $\overline{\mathbb{P}}_{\theta}^N$; this however is more cumbersome than will be needed and we will use the kernel in a less expensive manner. 
Note that it does not make sense to use the variables $(\overline{u}_1,\dots,\overline{u}_T)$
that are simulated inside $K_{l,\theta}$ to approximate $H_l$; this is because the marginal distribution
of such variables,  under the invariant measure of the Markov chain on an extended space \emph{is no longer}
$\overline{\mathbb{P}}_{\theta}^N$.
We remark the cost of applying
$K_{l,\theta}$ is $\mathcal{O}(\Delta_l^{-1}N^2T)$ in step 2.~of Algorithm  \ref{alg:cond_pf_0} and then
the rest of the algorithm has a cost of $\mathcal{O}(\Delta_l^{-1}NMT)$.

\begin{algorithm}[h]
\begin{enumerate}
\item{Input $l\in\mathbb{N}_0$ the level of discretization, $N\in\mathbb{N}$ the number of particles for the approach of Algorithm \ref{alg:basic_method},  $M$ the number of particles for the main method, $\theta\in\Theta$, 
and $(U_{1}',\dots,U_T')\in\mathsf{E}_l^T$. }
\item{Sample $(\overline{u}_1,\dots,\overline{u}_T)$ from $\overline{\mathbb{P}}_{\theta}^N$ via a sequential
application of Algorithm \ref{alg:basic_method}, with $N$ particles.}
\item{Initialize: Sample $U_1^i$ independently from $
\prod_{k=1}^{\Delta_{l}^{-1}} Q_{\mu_{(k-1)\Delta_l,\theta}^{l,N},\theta}^l(x_{(k-1)\Delta_l},dx_{k\Delta_l})$, 
$A_{0}^i=i$ for $i\in\{1,\dots,M-1\}$. Set $k=2$.}
\item{Sampling: for $i\in\{1,\dots,M-1\}$ sample $U_{k}^i|U_{k-1}^{A_{k-1}^i}$ using the kernel 
$$ 
\prod_{k=1}^{\Delta_{l}^{-1}} Q_{\mu_{t-1+(k-1)\Delta_l,\theta}^{l,N},\theta}^l(x_{t-1+(k-1)\Delta_l},dx_{t-1+k\Delta_l})
$$
where $x_{t-1}^i=x^{A_{k-1}^i}$.
Set $U_{k}^N=U_k'$ and for $i\in\{1,\dots,M-1\}$, $(U_{1}^i,\dots,U_k^i)=(U_1^{A_{k-1}^i},\dots,U_{k-1}^{A_{k-1}^i},U_{k}^i)$. If $k=T$ go to 4..}
\item{Resampling: Construct the probability mass function on $\{1,\dots,M\}$:
$$
r_1^i = \frac{G_{\theta}(x_k^i,y_k)}{\sum_{j=1}^MG_{\theta}(x_k^j,y_k)}.
$$
For $i\in\{1,\dots,M-1\}$ sample $A_k^i$ from $r_1^i$. Set $k=k+1$ and return to the start of 2..}
\item{Construct the probability mass function on $\{1,\dots,M\}$:
$$
r_1^i = \frac{G_{\theta}(x_T^i,y_T)}{\sum_{j=1}^M G_{\theta}(x_T^j,y_T)}.
$$
Sample $i\in\{1,\dots,M\}$ using this mass function and return $(U_1^i,\dots,U_T^i)$.}
\end{enumerate}
\caption{Conditional Particle Filter at level $l\in\mathbb{N}_0$.}
\label{alg:cond_pf_0}
\end{algorithm}

\subsection{Approximating Pairs of Smoothers}\label{sec:smooth_pair_approx}

We will now introduce a method that allows us to approximate the difference:
\begin{equation}\label{eq:ll_diff}
\nabla_{\theta}\log(p_{\theta}^{l,N_l}(y_1,\dots,y_T)) - 
\nabla_{\theta'}\log(p_{\theta'}^{l,N_{l-1}}(y_1,\dots,y_T))
\end{equation}
where the number of samples used in the particle approximations of the laws can be different as can be the parameter values.  Note as discussed before some further work is needed for calculating $H_l^N$,  so this is not the exact thing that is done in practice, but for now we proceed as if $H_l^N$ can be computed exactly.
We remark that the task of computing \eqref{eq:ll_diff} is of course possible to approximate this difference by using two independent runs associated to the kernels in 
Algorithm \ref{alg:cond_pf_0}.  However,  as we instrinsically will be relying on randomized multilevel Monte Carlo methods,  it is important that our estimators are suitably dependent;  as this is now well-understood in the literature we direct the reader to \cite{giles,giles1,hein,ml_rev} for a full explanation of why this should be the case.

The method that we describe concerns simulating a coupling of the pair of Markov kernels $(K_{l,\theta}(u_{1:T},\cdot),$ $K_{l-1,\theta'}(u'_{1:T},\cdot)$ where $u_{1:T}=(u_1,\dots,u_T)\in\mathsf{E}_l^T$,  $u_{1:T}'=(u_1',\dots,u_T')\in\mathsf{E}_{l-1}^T$.  If one can sample such a coupling,  then we can approximate \eqref{eq:ll_diff}
and this is exactly the subject of this section.  As we saw in Algorithm \ref{alg:cond_pf_0},  simulating the kernels
$K_{l,\theta}$ and $K_{l-1,\theta'}$ entails using Algorithm \ref{alg:basic_method} and as we seek to correlate the simulation,  we need a correlated version of Algorithm \ref{alg:basic_method}; this was solved in \cite{po_mv} and we present the same algorithm,  with notational changes in Algorithm \ref{alg:basic_method_coup}.

As in Section \ref{sec:mod_obj}, consider sequentially sampling Algorithm \ref{alg:basic_method_coup} from time
$1$ to time $T$ using the approximated laws from the previous time step,  to initialize the next time step; throughout $N_l,N_{l-1}$ is fixed and so is $\theta,\theta'$.
At time $t$,  write all the simulated random variables from Algorithm \ref{alg:basic_method_coup} as 
$\overline{u}_t\in\mathsf{E}_l^{N_l}$,  $\widetilde{u}_t\in\mathsf{E}_{l-1}^{N_{l-1}}$.
Finally write the joint probability of $(\overline{u}_1,\dots\overline{u}_T)$ and $(\widetilde{u}_1,\dots,\widetilde{u}_T)$  as $\overline{\mathbb{P}}_{\theta,\theta'}^{N_l,N_{l-1}}$.  This provides a means to sample a coupling of
$(\overline{\mathbb{P}}_{\theta}^{N_l},\overline{\mathbb{P}}_{\theta'}^{N_{l-1}})$ which we shall use in the 
the simulation of the coupling of $(K_{l,\theta}(u_{1:T},\cdot),K_{l-1,\theta'}(u'_{1:T},\cdot)$; this is the next task.

\begin{algorithm}[h]
\begin{enumerate}
\item{Input $l\in\mathbb{N}$ the level of discretization, $(N_l,N_{l-1})\in\mathbb{N}$ the number of particles
at levels $l$ and $l-1$ with $N_l\geq N_{l-1}$, $(\theta,\theta')\in\Theta^2$
,  $t\in\{1,\dots,T\}$. If $t=1$ set $\mu_{0,\theta}^{l,N_{l}}(dx)=\widetilde{\mu}_{0,\theta'}^{l-1,N_{l-1}}(dx)=\delta_{\{x_0\}}(dx)$ otherwise input a pair of empirical measures $\mu_{t-1,\theta}^{l,N_l}(dx)=\tfrac{1}{N_l}\sum_{i=1}^{N_l}\delta_{\{X_{t-1}^{l,i}\}}(dx)$, 
$\widetilde{\mu}_{t-1,\theta'}^{l-1,N_{l-1}}(dx)=\tfrac{1}{N_{l-1}}\sum_{i=1}^{N_{l-1}}\delta_{\{\widetilde{X}_{t-1}^{l-1,i}\}}(dx)$.  Set $k=1$.}
\item{For $i\in\{1,\dots,N_l\}$ generate:
\begin{align*}
X_{t-1+k\Delta_l}^{l,i} &  =  X_{t-1+(k-1)\Delta_l}^{l,i} + a_{\theta}\left(X_{t-1+(k-1)\Delta_l}^{l,i},\overline{\xi}_{\theta}(X_{t-1+(k-1)\Delta_l}^{l,i},\mu_{t-1+(k-1)\Delta_l,\theta}^{l,N_l})\right) + \\ & \sigma\left(X_{t-1+(k-1)\Delta_l}^{l,i}\right)\left[W_{t-1+k\Delta_l}^i - W_{t-1+(k-1)\Delta_l}^i\right]
\end{align*}
where
\begin{eqnarray*}
\overline{\xi}_{\theta}(X_{t-1+(k-1)\Delta_l}^i,\mu_{t-1+(k-1)\Delta_l}^{l,N_l}) & = & \frac{1}{N_l}\sum_{j=1}^{N_l} \xi_{\theta}(X_{t-1+(k-1)\Delta_l}^{l,i},X_{t-1+(k-1)\Delta_l}^{l,j})\\
\mu_{t-1+(k-1)\Delta_l,\theta}^{l,N_l}(dx) & = & \frac{1}{N_l}\sum_{j=1}^{N_l}\delta_{\{X_{t-1+(k-1)\Delta_l}^{l,j}\}}(dx).
\end{eqnarray*}
Set $k=k+1$, if $k=\Delta_l^{-1}+1$ go to step 3.~otherwise go to the start of step 2..}
\item{For $i\in\{1,\dots,N_{l-1}\}$ compute:
\begin{align*}
\widetilde{X}_{t-1+k\Delta_{l-1}}^{l-1,i} & =  \widetilde{X}_{t-1+(k-1)\Delta_{l-1}}^{l-1,i} + a_{\theta'}\left(\widetilde{X}_{t-1+(k-1)\Delta_{l-1}}^{l-1,i},\overline{\xi}_{\theta'}(\widetilde{X}_{t-1+(k-1)\Delta_{l-1}}^{l-1,i},\widetilde{\mu}_{t-1+(k-1)\Delta_{l-1},\theta'}^{l-1,N_{l-1}})\right) + \\ & \sigma\left(\widetilde{X}_{t-1+(k-1)\Delta_{l-1}}^{l-1,i}\right)\left[W_{t-1+k\Delta_{l-1}}^i - W_{t-1+(k-1)\Delta_{l-1}}^i\right]
\end{align*}
where
\begin{eqnarray*}
\overline{\xi}_{\theta'}(\widetilde{X}_{t-1+(k-1)\Delta_{l-1}}^{l-1,i},\widetilde{\mu}_{t-1+(k-1)\Delta_{l-1},\theta'}^{l-1,N}) & = & \frac{1}{N_{l-1}}\sum_{j=1}^{N_{l-1}} \xi_{\theta'}(\widetilde{X}_{t-1+(k-1)\Delta_{l-1}}^{l-1,i},\widetilde{X}_{t-1+(k-1)\Delta_{l-1}}^{l-1,j})\\
\widetilde{\mu}_{t-1+(k-1)\Delta_{l-1},\theta'}^{l-1,N_{l-1}}(dx) & = & \frac{1}{N_{l-1}}\sum_{j=1}^{N_{l-1}}\delta_{\{\widetilde{X}_{t-1+(k-1)\Delta_l}^{l-1,j}\}}(dx)
\end{eqnarray*}
and the increments of the Brownian motion $\left[W_{t-1+k\Delta_{l-1}}^i - W_{t-1+(k-1)\Delta_{l-1}}^i\right]$ were generated in step 2..
Set $k=k+1$, if $k=\Delta_{l-1}^{-1}+1$ go to step 4.~otherwise go to the start of step 3..}
\item{Output all the required laws $\mu_{t-1+\Delta_l,\theta}^{l,N_l},\dots,\mu_{t,\theta}^{l,N_l}$,  $\widetilde{\mu}_{t-1+\Delta_l,\theta'}^{l-1,N_{l-1}},\dots,\widetilde{\mu}_{t,\theta'}^{l-1,N_{l-1}}$.}
\end{enumerate}
\caption{Approximating the Consecutive Laws when starting with a particle approximation at time $t-1$, $t\in\{1,\dots,T\}$.}
\label{alg:basic_method_coup}
\end{algorithm}

In order to specify the simulation of our coupling of  $(K_{l,\theta}(u_{1:T},\cdot),K_{l-1,\theta'}(u'_{1:T},\cdot)$, we need to further ingredients. The first is the ability to sample the maximum coupling of two positive probability mass functions on a finite set $\{1,\dots,M\}$. This is easily achieved and is detailed in several places 
and we refer the reader to e.g.~\cite{ub_grad_new,po_mv}.  The second is the simulation of a synchronous coupling of two first order Euler-Maruyama time discretizations of 
$$
 \overline{P}_{\mu_{t-1,\theta}^{l,N_l},t,\theta}^l(x_{t-1},du_t)
\quad \textrm{and} \quad
 \overline{P}_{\mu_{t-1,\theta}^{l-1,N_{l-1}},t,\theta'}^{l-1}(x_{t-1},du_t)
$$
where the input probability measures are simulated via Algorithm \ref{alg:basic_method_coup} independently
of all other random variables represented in the displayed equation. This has,  in effect, been described in 
Algorithm \ref{alg:basic_method_coup},  except of course we need only one particle ($N_l=N_{l-1}=1$ in Algorithm \ref{alg:basic_method_coup}) and the empirical measures used to approximate the law of the MV SDE need not be updated,  as they have been approximated and plugged in already.  Therefore we will write such a simulation of $(u_t^l,u_{t}^{l-1})\in\mathsf{E}_l\times\mathsf{E}_{l-1}$
conditonal upon $(u_{t-1}^l,u_{t-1}^{l-1})$ (where $(u_{0}^l,u_{0}^{l-1})=(x_0,x_0))$ as the kernel $\overline{P}^{l,l-1}_{\theta,\theta'}$.  The simulation of the coupling of $(K_{l,\theta}(u_{1:T},\cdot),K_{l-1,\theta'}(u'_{1:T},\cdot)$ is decribed in Algorithm \ref{alg:ccpf} and will be referred to as $\check{K}_{l,\theta,\theta'}(z_{1:T},\cdot)$.    
$z_{1:T}=(u_{1:T},u_{1:T}')\in\mathsf{E}_l^T\times\mathsf{E}_{l-1}^{T}$.  The method is simply a conditional version of the coupled particle filter in \cite{po_mv}.
As before, marginally,  the invariant measure in the level $l$ collection of variables is $\overline{\pi}_{\theta}^{l,N_l}$ and level $l-1$ is $\overline{\pi}_{\theta'}^{l-1,N_{l-1}}$.
The cost of simulation of $\overline{\mathbb{P}}_{\theta,\theta'}^{N_l,N_{l-1}}$ is $\mathcal{O}(N_l^2\Delta_{l}^{-1}T)$ and then of the subsequent parts of Algorithm \ref{alg:ccpf} as $\mathcal{O}(N_lM\Delta_{l}^{-1}T)$.

\begin{algorithm}[h]
\begin{enumerate}
\item{Input $l\in\mathbb{N}$ the level of discretization, $(N_l,N_{l-1})\in\mathbb{N}^2$, $N_l\geq N_{l-1}$, the number of particles for the approach of Algorithm \ref{alg:basic_method_coup},  $M$ the number of particles for the main method, $(\theta,\theta')\in\Theta^2$, 
and $(U_{1}^{',l},\dots,U_T^{',l})\in\mathsf{E}_l^T$,  $(U_{1}^{',l-1},\dots,U_T^{',l-1})\in\mathsf{E}_{l-1}^T$ }
\item{Sample $(\overline{u}_1,\dots,\overline{u}_T)$ 
and $(\widetilde{u}_1,\dots,\widetilde{u}_T)$
from $\overline{\mathbb{P}}_{\theta,\theta'}^{N_l,N_{l-1}}$ via a sequential
application of Algorithm \ref{alg:basic_method_coup}, with $N_l,N_{l-1}$ particles.}
\item{Set $k=1$, $A_{0}^{i,l}=A_0^{i,l-1}=i$ for $i\in\{1,\dots,M-1\}$.}
\item{Sampling: for $i\in\{1,\dots,M-1\}$ sample $(U_{k}^{i,l},U_k^{i,l-1})|(U_{k-1}^{A_{k-1}^{i,l},l},U_{k-1}^{A_{k-1}^{i,l-1},l-1})$ using the Markov kernel $\overline{P}^{l,l-1}_{\theta,\theta'}$. Set $U_{k}^{N,l}=U_k^{',l}$,  $U_{k}^{N,l-1}=U_k^{',l-1}$ and for $(i,s)\in\{1,\dots,M-1\}\times\{l-1,l\}$, $U_{0:k}^{i,s}=(U_0^{A_{k-1}^{i,s},s},\dots,U_{k-1}^{A_{k-1}^{i,s},s},U_{k}^{i,s})$. If $k=T$ go to 4..}
\item{Resampling: Construct the probability mass functions on $\{1,\dots,N\}$:
$$
r_1^i = \frac{G_{\theta}(x_k^{i,l},y_k)}{\sum_{j=1}^MG_{\theta}(x_k^{j,l},y_k)} \quad\textrm{and}\quad
r_2^i = \frac{G_{\theta'}(x_k^{i,l-1},y_k)}{\sum_{j=1}^MG_{\theta'}(x_k^{j,l-1},y_k)}.
$$
For $i\in\{1,\dots,M-1\}$ sample $A_k^{i,l}$ and $A_k^{i,l-1}$ using the maximal coupling of $r_1^i$ and $r_2^i$. Set $k=k+1$ and return to the start of 2..}
\item{Construct the probability mass functions on $\{1,\dots,M\}$:
$$
r_1^i = \frac{G_{\theta}(x_T^{i,l},y_T)}{\sum_{j=1}^MG_{\theta}(x_T^{j,l},y_T)} \quad\textrm{and}\quad
r_2^i = \frac{G_{\theta'}(x_T^{i,l-1},y_T)}{\sum_{j=1}^MG_{\theta'}(x_T^{j,l-1},y_T)}
$$
Sample $(i,j)\in\{1,\dots,M\}^2$ using these mass functions via the maximal coupling and return $(U_{1:T}^{i,l},,U_{1:T}^{i,l-1})$.}
\end{enumerate}
\caption{Coupled Conditional Particle Filter at level $l\in\mathbb{N}$.}
\label{alg:ccpf}
\end{algorithm}

\subsection{Stochastic Approximation}\label{sec:sa}

We will now detail the Markovian stochastic approximation (MSA) method that was adopted in \cite{ub_grad_new}
and adapted from \cite{andr} which can be used to compute the optimizer of $\log(p_{\theta}^{l,N}(y_1,\dots,y_T))$. We note that there are several user set parameters that are needed and a discussion is deferred to Sections \ref{sec:disc_meth} and  \ref{sec:theory}.  The calculation of $H_l^N$ is explained in details in Section \ref{sec:hl_comp}. We will use a sequence $\{\gamma_n\}_{n\in\mathbb{N}}$ of positive real numbers that is square summable but a divergent sum (i.e.~$\sum_{n=1}^{\infty}\gamma_n=\infty$).

The basic method is as follows,  which we shall refer to as Procedure 1 at level $l$.
\begin{enumerate}
\item{Set $\theta^{l,0}\in\Theta$ and approximate the laws of MV SDE using $\overline{\mathbb{P}}_{\theta^{l,0}}^{N_l}$ and
generate $U_{1:T}^{l,0}$ using
$$
\prod_{t=1}^T  \overline{P}_{\mu_{t-1,\theta^{l,0}}^{l,N_l},t,\theta^{l,0}}^l(x_{t-1},du_t)
$$
where $\mu_{0,\theta^{l,0}}^{l,N_l}(dx)=\delta_{\{x_0\}}(dx)$.
Set $n=1$.}
\item{Sample $U_{1:T}^{l,n}|(\theta_{l,0},u_{1:T}^{l,0}),\dots,(\theta^{l,n-1},u_{1:T}^{l,n-1})$ from $K_{l,\theta^{l,n-1}}(u_{1:T}^{l,n-1},\cdot)$.} 
\item{Update:
$$
\theta^{l,n} = \theta^{l,n-1} + \gamma_n H_l^{N_l}(\theta_{n-1}^l,s_{1:T}^{l,n})
$$
where in order to compute $H_l^{N_l}$ one has approximated the laws of the MV SDE using a new sample from  $\overline{\mathbb{P}}_{\theta^{l,n-1}}^{N_l}$.
Set $n=n+1$ and go to the start of 2..}
\end{enumerate}
We note that we do not specify a stopping iteration (time),  but in practice the algorithm will run for a finite number of iterations. The cost of one iteration of Procedure 1 is $\mathcal{O}(\{N_l^2\Delta_l^{-1} + N_lM\Delta_l^{-1}\}T)$ where $M$
is the number of samples used in Algorithm \ref{alg:cond_pf_0}. 

We will also need a method to generate couples of parameters,  which we now describe and
is referred to as Procedure 2 at level $l$.
\begin{enumerate}
\item{Set $(\theta^{l,0},\theta^{l-1,0})\in\Theta^2$, approximate the laws of MV SDEs using $\overline{\mathbb{P}}_{\theta^{l,0},\theta^{l-1,0}}^{N_l,N_{l-1}}$ and
generate $Z_{1:T}^{l,0}=(U_{1:T}^{l,0},U_{1:T}^{l-1,0})$ using a coupling of 
$$
\prod_{t=1}^T\overline{P}_{\mu_{t-1,\theta^{l,0}}^{l,N_l},t,\theta^{l,0}}^l(x_{t-1},du_t)
\quad
\textrm{and}
\quad
\prod_{t=1}^T \overline{P}_{\mu_{t-1,\theta^{l-1,0}}^{l-1,N_{l-1}},t,\theta^{l-1,0}}^l(x_{t-1},du_t)
$$
where we recall that $\mu_{0,\theta^{l,0}}^{l,N_l}(dx)=\mu_{0,\theta^{l-1,0}}^{l-1,N_{l-1}}(dx)=\delta_{\{x_0\}}(dx)$.
Set $n=1$.}
\item{Sample $Z_{1:T}^{l,n}|(\theta_{l,0},\theta^{l-1,0},z_{1:T}^{l,0}),\dots,(\theta^{l,n-1},\theta^{l-1,n-1},z_{1:T}^{l,n-1})$ from $\check{K}_{l,\theta^{l,n-1},\theta^{l-1,n-1}}(z_{1:T}^{l,n-1},\cdot)$.} 
\item{Update:
\begin{eqnarray*}
\theta^{l,n} & = & \theta^{l,n-1} + \gamma_n H_l^{N_l}(\theta^{l,n-1},s_{1:T}^{l,n}) \\
\theta^{l-1,n} & = & \theta^{l-1,n-1} + \gamma_n H_{l-1}^{N_{l-1}}(\theta^{l-1,n-1},s_{1:T}^{l-1,n})
\end{eqnarray*}
where in order to compute $H_l^{N_l},H_{l-1}^{N_{l-1}}$ one has approximated the laws of the MV SDE using a new sample from
$\overline{\mathbb{P}}_{\theta^{l,n-1},\theta^{l,n-1}}^{N_l,N_{l-1}}$.
Set $n=n+1$ and go to the start of 2..}
\end{enumerate}
As for Procedure 1,  we do not set a stopping iteration, but of course there will be one used in practice.
Similarly to Procedure 1, the cost of one iteration of Procedure 2 is $\mathcal{O}(\{N_l^2\Delta_l^{-1} + N_lM\Delta_l^{-1}\}T)$ where $M$
is the number of samples used in Algorithm \ref{alg:ccpf}. 

\subsubsection{Computation of $H_l^N$}\label{sec:hl_comp}

Recall the expression for $H_l$ in \eqref{eq:Hl_def},  which features computations of gradients associated to $a_{\theta}$ the drift of \eqref{eq:sde}.  As the drift depends on an expectation w.r.t.~the law of the MV SDE,  this naturally means that one could track a gradient associated to the the law of the MV SDE (or an approximation thereof).  This would further complicate what is already a collection of algorithms which are already quite advanced.
As a result, when needed at a level $l$,  we will replace such derivatives with a finite-difference approximation with step size exactly $\Delta_l$.  For instance if one needs to calculate for some co-ordinate $\theta^{(j)}$ of $\theta$,
$j\in\{1,\dots,d_{\theta}\}$
$$
\frac{\partial}{\partial \theta^{(j)}}\left\{\overline{\xi}_{\theta}(x,\mu_{t,\theta}^{l,N_l})\right\}
$$
then it is replaced with a finite difference approximation where the laws needed have been approximated using
$\overline{\mathbb{P}}_{\theta^{l}}^{N_l}$ or $\overline{\mathbb{P}}_{\theta^{l},\theta^{l}}^{N_l,N_{l-1}}$ and if needed rerun for the perturbation in the finite difference.  That is,  the approximation in this given case is
$$
\frac{1}{\Delta_l}\left(\overline{\xi}_{\theta}(x,\mu_{t,\theta}^{l,N_l})-\overline{\xi}_{\theta(j,l)}(x,\mu_{t,\theta(j,l)}^{l,N_l})\right)
$$
where $\theta(j,l)=(\theta^{(1)},\dots,\theta^{(j-1)},\theta^{(j)},\theta^{(j+1)}-\Delta_l,\dots,\theta^{(d_{\theta})})^{\top}$ and the laws have been simulated from $\overline{\mathbb{P}}_{\theta(l,j)}^{N_l}$ (either directly or marginally) to obtain
$\mu_{t,\theta(j,l)}^{l,N}$.  In reference to the simulation of the perturbed 
$\overline{\mathbb{P}}_{\theta(l,j)}^{N_l}$ if we are in the context of Procedure 1,  then it is a direct sample from
$\overline{\mathbb{P}}_{\theta(l,j)}^{N_l}$, otherwise we sample from $\overline{\mathbb{P}}_{\theta^{l}(j,l),\theta^{l-1}(j,l)}^{N_l,N_{l-1}}$. We write the resulting $H_l^N$ as $\widehat{H}_l^N$ from herein.
Whilst this is far from optimal, requiring $d_{\theta}-$extra simulations from $\overline{\mathbb{P}}_{\theta(l,j)}^{N_l}$, as remarked already,  this leads to minimal additional complications in Procedures 1 and 2.  We also note that generally $d_{\theta}$ is typically only a few dimensions at least in our applications.

\subsection{Final Algorithm}\label{sec:fa}

Our methodology for estimating the parameters of our partially observed MV SDE model is now described.
We use a randomized multilevel Monte Carlo method (see e.g.~\cite{rhee,matti}) which will require a positive probability mass function (PMF),  $\mathbb{P}_L$,  on the level of time-discretization which we constrain to the set $\{0,\dots,L\}$ and a positive PMF,  $\mathbb{P}_P$, on the number of iterations of the MSA approaches in Procedures 1 and 2.  
Additionally we need a sequence of increasing positive integers,  that go to infinity which we shall denote
as $\{T_p\}_{\in\mathbb{N}_0}$.
For now we shall assume that the afore-mentioned have been specified.  For the final algorithm, we note that typically one can run mutiple versions of this algorithm in parallel and average the estimators.
The final algorithm is now presented.

\begin{enumerate}
\item{Sample $l$ from $\mathbb{P}_L$ and $p$ from $\mathbb{P}_P$.}
\item If $l=0$ perform the following:
\begin{itemize}
\item{Run Procedure 1 at level 0 for $T_p$ iterations.}
\item{If $p=0$ return
$$
\widehat{\theta}_{\star} = \frac{\theta^{l,T_p}}{\mathbb{P}_P(p)\mathbb{P}_L(l)},
$$
otherwise  return
$$
\widehat{\theta}_{\star} = \frac{\theta^{l,T_p}-\theta^{l,T_{p-1}}}{\mathbb{P}_P(p)\mathbb{P}_L(l)}.
$$
}
\end{itemize}
\item{If $l\in\{1,\dots,L\}$, run Procedure 2 at level $l$ for $T_p$ iterations.}
\begin{itemize}
\item If $p=0$ return
$$
\widehat{\theta}_{\star} = \frac{\theta^{l,T_p}-\theta^{l-1,T_p}}{\mathbb{P}_P(p)\mathbb{P}_L(l)},
$$
otherwise  return
$$
\widehat{\theta}_{\star} = \frac{\theta^{l,T_p}-\theta^{l-1,T_{p}}-\{\theta^{l,T_{p-1}}-\theta^{l-1,T_{p-1}}\}}{\mathbb{P}_P(p)\mathbb{P}_L(l)}.
$$
\end{itemize}
\end{enumerate}

\subsection{Discussion of Methodology}\label{sec:disc_meth}

The approach that we developed in the previous section is a substantial extension of the approach in \cite{ub_grad_new} (see also \cite{ub_pf} for a related method).  We now discuss several aspects of the methodology for clarifications and also to compare with several ideas that exist in the literature.

One of the first points of interest is of course what it is that we will hope to compute via our estimator
$\widehat{\theta}_{\star}$ via \eqref{eq:mod_grad_ll}.  
Let $\check{\mathbb{E}}_{\theta}^{N_L}$ denote the expectation associated to 
$\overline{\mathbb{P}}_{\theta}^{N_L}$ and the
extra-randomness needed to compute $\widehat{H}_L^{N_L}$ (see Section \ref{sec:hl_comp}) and $\theta_{\star}^L$ the unique value
which makes
$$
\widehat{\nabla_{\theta}\log(p_{\theta}^{L,N_L}(y_1,\dots,y_T))} := 
\int_{\mathsf{E}_L^T}
\check{\mathbb{E}}_{\theta}^{N_L}\left[\widehat{H}_L^{N_L}(\theta,s_1,\dots,s_T)\right]
\overline{\pi}^{L,N_L}_{\theta}\left(d(u_1,\dots,u_T)\right)
$$
zero,  which is assumed exist,  then in the next section we will show that
$\mathbb{E}[\widehat{\theta}_{\star}]=\theta_{\star}^L$ where we use
$\mathbb{E}[\cdot]$ to denote the expectation w.r.t.~the probability law associated to the final algorithm.
Using the theory of \cite{basic_method}, Euler-Maruyama and finite difference methods,  one expects that under assumptions
$\theta_{\star}^L\rightarrow\theta^{\star}$ the (assumed) unique maximizer of the continuous-time likelihood.
So although one will have a bias,  there is an expectation that this should be small especially if $L$ is sufficiently large; for this particular model,  the subject of the rate of decrease of the bias as $L$ grows is to the best of our knowledge unstudied in the literature.

There are several user-set parameters in the algorithm and we try to give some guidance on the selection of these here and in Section \ref{sec:theory}.  First is (non-decreasing) the sequence $\{N_l\}_{l\in\{0,\dots,L\}}$ which are used in the approximation of the laws in $H_l$ and also in the two MCMC methods of Algorithms \ref{alg:cond_pf_0} (called $N$ there) and \ref{alg:ccpf}.  The approximation of the laws is critical to accuracy of the methodology and as such we have increased the computational effort with the level of time-discretization. 
This is of course not the only possibility and one option would be to decouple the time-discretization and the accuracy of the laws,  and then,  instead of focussing on a multilevel method,  consider a multi-index type approach \cite{haji},  which in-fact has been considered in a much different context in \cite{nadir}.  This possibility is not considered in this article but is a viable alternative that could reduce the cost of the approach that we use.  In terms of the actual selection of the sequence,  this will be considered in the next section. 
Second the parameter $M$ that is in Algorithms \ref{alg:cond_pf_0} and \ref{alg:ccpf} is well-studied in the literature,  see for instance \cite{andrieu,andr} and we use the convention that $M=\mathcal{O}(T)$. The choices
of  $\mathbb{P}_L$,  $\mathbb{P}_P$ and $\{T_p\}_{\in\mathbb{N}_0}$ are discussed below.

\subsection{Theory}\label{sec:theory}

We now start by giving the main mathematical result of the article.   It should first be noted that our proof applies to a version of our estimator that is subject to reprojection which has been explained in several articles; see \cite{andr3,andr,ub_grad_new}.  As this approach is never implemented in our numerical results,  we do not make further mention of it here and refer the reader to those articles and the references therein.  We also assume that the matrix $\sigma(x)$ in \eqref{eq:sde} is diagonal,  which simplifies many calculations.  In addition to this several
assumptions are made in Appendix \ref{app:ass} which can be seen there with a discussion.  We can now state our main result.  The proof of the following result is contained within the entirety of Appendix \ref{app:theory} but with the final argument given in Appendix \ref{app:main_theo}.

\begin{theorem}\label{theo:main_thm}
Assume (A\ref{ass:1_new}-\ref{ass:4}).  Then we have that $\mathbb{E}[\widehat{\theta}_{\star}]=\theta_{\star}^L$ .
\end{theorem}

The result is essentially the minimal one that one could hope for.  It says nothing about whether the estimator is of finite variance or how one can select some of the remaining simulation parameters.  To this end, based upon the results in \cite{po_mv,frika,ub_pf} we \emph{conjecture} that, under appropriate assumptions, the variance of our estimator is upper-bounded by an expression that is
\begin{equation}\label{eq:var_conj}
\mathcal{O}\left(
\frac{1}{\mathbb{P}_L(0)}\sum_{p=0}^{\infty}\frac{1}{\mathbb{P}_P(p)T_p} + 
\sum_{l=1}^L\sum_{p=0}^{\infty}\frac{1}{\mathbb{P}_L(l)\mathbb{P}_P(p)}
\left\{
\frac{\Delta_l^{1/2}}{T_p} +
\frac{L\Delta_l}{N_l}
\right\}
\right).
\end{equation}
Just as in \cite{ub_pf} it is difficult to find a combination of simulation parameters which ensure that the expected cost and \eqref{eq:var_conj} is simultaneously finite.  However, one can choose (for instance as in \cite{ub_pf}) $T_p=2^p$,
$N_l=lL$, $\mathbb{P}_L(l)\propto \Delta_l^{1/2\varepsilon}$,  $\varepsilon\in(0,1)$,  and $\mathbb{P}_P(p)\propto 2^p(p+1)\log_2(p+2)$ and the expression in \eqref{eq:var_conj} is finite.  As noted in \cite{ub_pf} this also implies that the cost is finite with high-probability.   Note that the presence of $L$ in \eqref{eq:var_conj} (which could possibly removed \cite{po_mv}) is one reason why we truncate the level of time discretization in our estimator; hence we have a randomized multilevel method and not an unbiased method.  We also remark that to verify \eqref{eq:var_conj} is likely to be particularly arduous,  see for instance the proofs in \cite{ub_grad} in a simpler context and so we leave this to future work.

\section{Numerical Simulations}\label{sec:numerics}

In this section we validate our methodology and algorithms by applying them to two distinct models: one from computational neuroscience and another from the domain of deep learning. These models serve as rigorous test cases, allowing us to assess the robustness and efficacy of our approach across diverse application areas.

\subsection{Computational Neuroscience: Continuous-Time Stochastic Kuramoto }

\subsubsection{Model}

The continuous-time stochastic Kuramoto model extends the classical Kuramoto model used in computational neuroscience, combining both deterministic and stochastic components to better capture phase synchronization dynamics within networks of oscillators \cite{Kuramoto}. This model is crucial for analyzing collective behaviors in neural networks and other oscillatory systems where interactions are governed by mean-field theory. By integrating mean-field interactions with stochastic noise, the model provides a robust framework for studying synchronization phenomena. It effectively addresses the complexities of phase alignment in coupled oscillators, making it a valuable tool for understanding emergent properties such as synchronization and desynchronization in neural systems and various other oscillatory networks.

For any time \( t \in [0,T] \), the model is described by the following dynamics of the one-dimensional McKean-Vlasov stochastic differential equation:
\begin{equation}
    dX_t = \left(\theta + \int \sin(X_t - x) \, \mu_t(dx)\right) dt + \sigma \, dW_t, \quad X_0 = x_0 \in \mathbb{R},
    \label{Kur}
\end{equation}
where $\{X_t\}_{t\in[0,T]}$ represents the phase of the oscillator at time \( t \) for a neuron in the network, and \( \theta \) is the intrinsic frequency of the oscillator. The term  $\int \sin(X_t - x) \, d\mu_t(dx)$ captures the mean-field interaction, with $\mu_t(y)$ denoting the law of the phases of all other oscillators in the network at time \( t \). Finally, the diffusion coefficient  $\sigma > 0$  indicates the intensity of the stochastic noise term, which accounts for random fluctuations or external disturbances affecting the oscillator's phase, while $\{W_t\}_{t\in[0,T]}$ denotes the one-dimensional Wiener process that introduces randomness into the phase dynamics.

We numerically simulate the dynamics of our first model (\ref{Kur}) with parameters including a frequency (\(\theta = 0.50\)), noise intensity (\(\sigma = 0.15\)), and 100 oscillators over a time span of 100 units. The simulation captures the phase dynamics of these oscillators, accounting for both mean-field interactions and stochastic noise.
We have the phase evolution (Figure \ref{fig: SimKuramoto}(a)) which displays how the phases of all oscillators evolve over time, showing trends toward synchronization despite noise.
The single oscillator phase (Figure \ref{fig: SimKuramoto}(b)) illustrates the trajectory of one oscillator's phase, highlighting its interaction with the group.

\begin{figure}[H]
\centering
\subfloat[]{\includegraphics[width=0.45\textwidth]{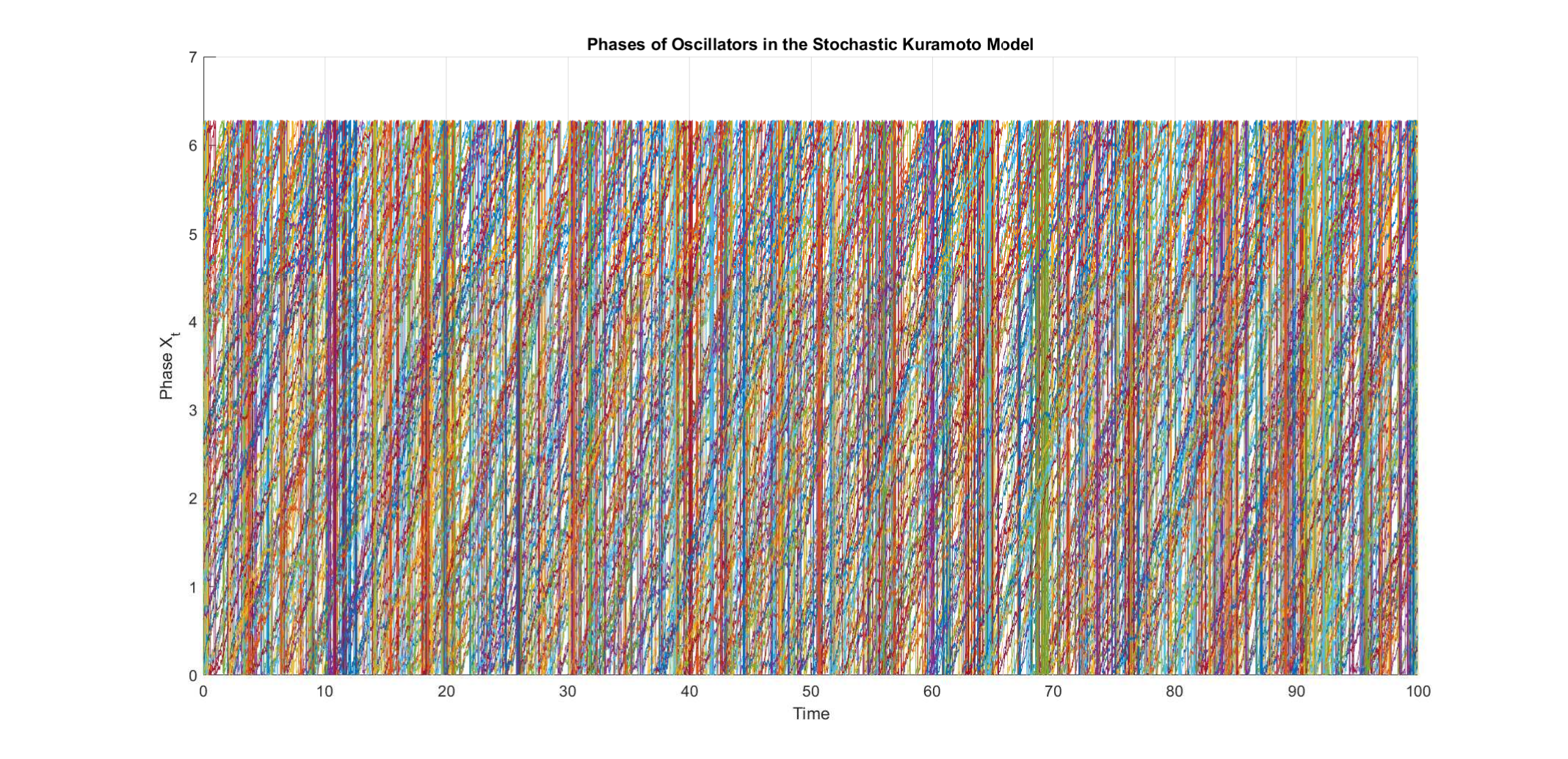}}\qquad
\subfloat[]{\includegraphics[width=0.45\textwidth]{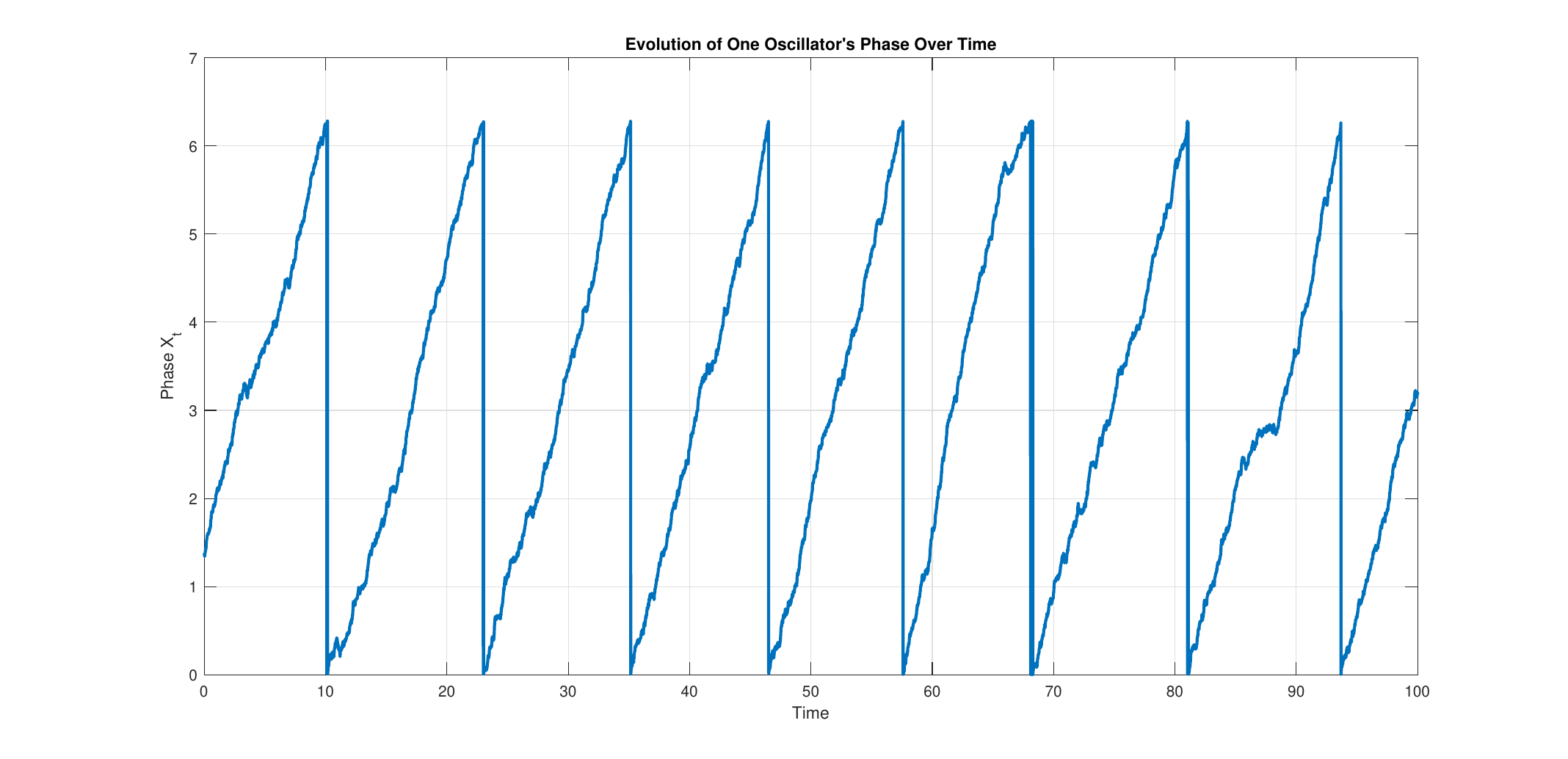}}\qquad
\caption{Outputs of the stochastic Kuramoto model. (a) Phases of the oscillators over time. (b) Evolution of one oscillator's phase over time.}
\label{fig: SimKuramoto}
\end{figure}

\subsection{Simple Training Neural Networks with Mean-Field Interactions}

Training neural networks with McKean-Vlasov stochastic differential equations use mean-field interactions to model how each neuron's activation depends on both its state and the collective behavior of the network. This approach improves the scalability and effectiveness of deep learning models by accounting for network-wide influences and handling noise and uncertainty better. Ongoing research is expected to yield new techniques for training neural networks and broaden their applications in artificial intelligence \cite{DeepL1}. Below, we present an example of a feedforward neural network architecture.
\begin{center}

\begin{tikzpicture}[node distance=2cm, font=\footnotesize]

    \node[input] (i1) {Input 1};
    \node[input, below of=i1] (i2) {Input 2};
    \node[input, below of=i2] (i3) {Input 3};
    
    \node[hidden, right of=i1, xshift=3cm] (h1) {Hidden 1};
    \node[hidden, below of=h1] (h2) {Hidden 2};
    \node[hidden, below of=h2] (h3) {Hidden 3};
    
    \node[output, right of=h1, xshift=3cm] (o1) {Output 1};
    \node[output, below of=o1] (o2) {Output 2};
    
    \foreach \i in {1,2,3}
        \foreach \h in {1,2,3}
            \draw[arrow] (i\i) -- (h\h);
    \foreach \h in {1,2,3}
        \foreach \o in {1,2}
            \draw[arrow] (h\h) -- (o\o);

\end{tikzpicture}

\end{center}

\subsubsection{Model}

The evolution of the activation $X_t$ of a neuron in a neural network can be described by the following 
$d-$dimensional system of
McKean-Vlasov SDEs, for $j\in\{1,\dots,d\}$
\begin{equation}
dX_t^{(j)} = \left\{\alpha\left(\frac{1}{d}\sum_{j=1}^d \int f(x)\mu_t^{(j)}(dx) - f(X_t^{(j)})\right)
+ \beta \left( \overline{w} X_t^{(j)} - b\right)
\right\}dt + \sigma dW_t^{(j)}
\label{NN}
\end{equation}
where $\mu_t^{(j)}$ is the marginal of the law of $X_t$ in co-ordinate $j$,  $f:\mathbb{R}\rightarrow\mathbb{R}$
is an activation function (explained below) and $\{W_t\}_{t\in[0,T]}$ is $d-$dimensional standard Brownian motion. 
$(\alpha,\beta,b,\overline{w},\sigma)\in\mathbb{R}^4\times\mathbb{R}^+$ are parameters which are now discussed.

In the model \eqref{NN},  $X_t$ denotes the activation of a representative neuron at time $t$, capturing the network's overall dynamics through a mean-field interaction with other neurons. The activation function $f(\cdot)$, such as ReLU, sigmoid, or hyperbolic tangent, introduces non-linearity, which is crucial for capturing complex patterns within the data and enhancing learning dynamics. The term 
$\frac{1}{d}\sum_{j=1}^d \int f(x) \mu_t^{(j)}(dx)$ represents the mean activation across neurons, aligning each neuron's behavior with the collective network state.  The parameters $\alpha$ and $\beta$ respectively control the influence of the mean-field interaction and the rate of change in $X_t$ due to the neuron's weighted input $\overline{w} X_t^{(j)}$ and bias $b$. The noise term, scaled by $\sigma$, introduces stochasticity,  where $W_t^{(j)}$ represents a Wiener process adding random fluctuations to the neuron's activation dynamics.

We numerically simulate the dynamics of this model (\ref{NN}) using the following parameter values: coupling strength $\alpha = 0.5$, rate of change in activation $\beta = 0.35$, noise intensity $\sigma = 0.15$, with the activation function $f(\cdot)$ set as a sigmoid function. For this simulation, we consider a total of $d = 100$ neurons over a time span of $T = 100$ seconds with a time step of $0.01$ seconds. These values provide a framework for analyzing neuron activation dynamics within the network, as illustrated in Figure \ref{fig: SimDeepL}.

\begin{figure}[H]
\centering
\subfloat[]{\includegraphics[width=0.45\textwidth]{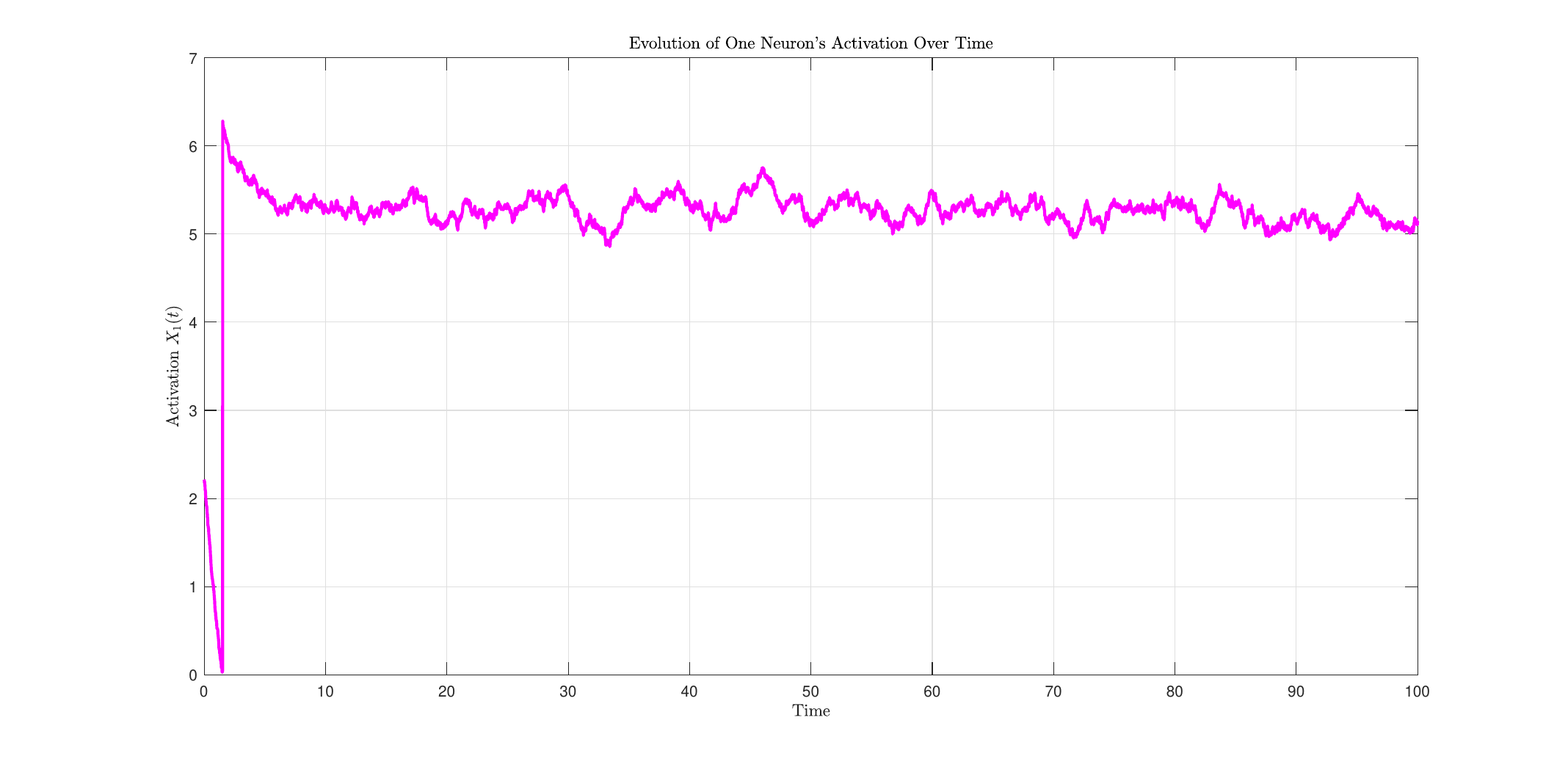}}\qquad
\subfloat[]{\includegraphics[width=0.45\textwidth]{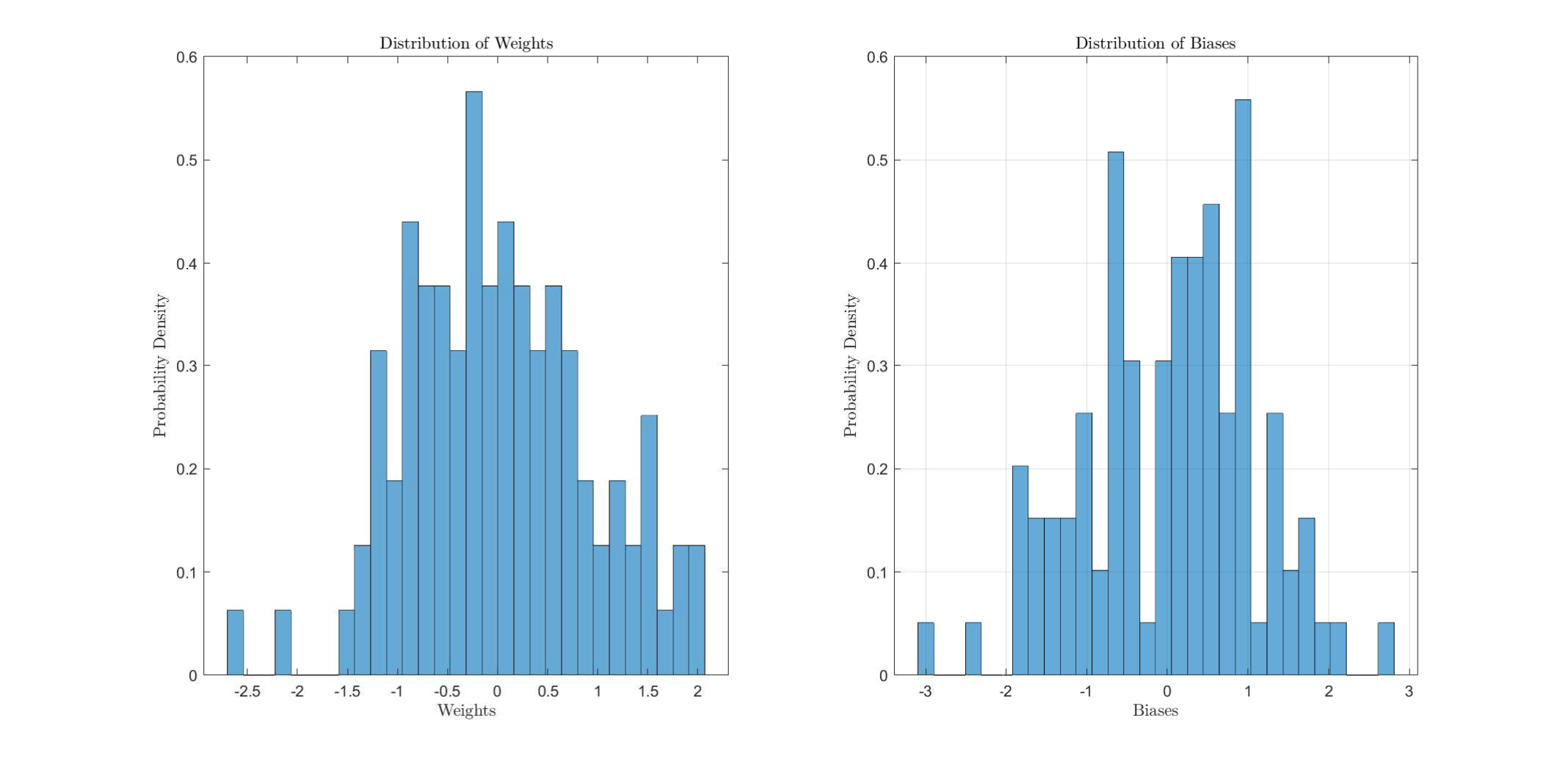}}
\caption{Outputs of the training neural networks with McKean-Vlasov stochastic differential equations. (a) Evolution of one neuron's activation over time $X_1$. (b) Histograms of weights and biases.}
\label{fig: SimDeepL}
\end{figure}

\subsection{Numerical Settings and Results}
In our simulations of both McKean-Vlasov dynamics models, we analyze observations denoted as  $\{ Y_k \}_{k=1}^{T}$, which are recorded at unit time intervals. These observations are modeled according to the conditional distribution $ Y_k | X_k = x_k \sim \mathcal{N}(x_k, \tau^2)$  (one-dimensional normal distribution of mean $x_k$ and variance $\tau^2$) for each $ k \in \{1, \ldots, T\}$,  where \( X_k \) represents the underlying state variable at time \( k \) and the noise parameter \( \tau \) quantifies the variance of the observations. In the context of the Kuramoto model, which describes coupled oscillators, our primary objective is to estimate the parameter of interest: the frequency of the oscillators denoted as \( \theta \). The model is initialized with the starting phase of the oscillator set at \( x_0 = 1.5 \), and the simulations are run for a total time period of \( T = 100 \) units. The parameters governing the dynamics include \( \sigma \), representing the intensity of the inherent stochasticity in the model, which is set to \( 0.15 \), and \( \tau \), fixed at \( 1 \), representing the noise level in the observations. These parameters collectively influence the behavior of the oscillator system and the accuracy of the estimated frequency \( \theta \). For the Neural Network McKean-Vlasov model, our aim shifts toward estimating two critical parameters: the coupling strength \( \alpha \) and the rate of change in activation \( \beta \). This model leverages the principles of neural networks to capture the dynamics of the system through its coupling and activation mechanisms. Notably, the values of \( \sigma \) and \( T \) remain identical to those used in the Kuramoto model, reinforcing the comparability between the two models and ensuring that any differences in estimated parameters arise from the distinct underlying dynamics rather than variations in the simulation settings. 

We establish a discretization scheme by setting the minimum discretization level \( l_{\text{min}} = 2 \). For each level \( l \geq 2 \), we define a step size \( \Delta_l = 2^{-l} \), which allows finer temporal resolutions as \( l \) increases.
To obtain robust statistical estimates, we run our algorithm in parallel for $\overline{M}$ multiple repetitions (different from $M$ used e.g.~in Algorithm \ref{alg:cond_pf_0}).
In these simulations, the parameters of interest we aim to estimate are denoted by  $\widehat{\theta}_{\star}^1, \dots, \widehat{\theta}^{\overline{M}}_{\star}$ , where each  $\widehat{\theta}^j_{\star}$ represents an individual estimate from a single run of the final algorithm.  By averaging these estimates, we define an aggregated estimator $ \overline{\theta}^{\overline{M}}_{\star} = \frac{1}{\overline{M}} \sum_{j=1}^{\overline{M}} 
\widehat{\theta}^j_{\star} $. This averaged estimator $\overline{\theta}^{\overline{M}}_{\star}$ provides a more stable and accurate representation of the underlying parameter, reducing the variance associated with any single run and leveraging the law of large numbers.
To evaluate the accuracy of our aggregated estimator  $\overline{\theta}^{\overline{M}}_{\star}$, we compute the mean squared error (MSE), which measures the deviation of $\overline{\theta}^{\overline{M}}_{\star}$ from the true parameter value. The MSE is calculated using the formula:
$$
\text{MSE} = \frac{1}{100} \sum_{i=1}^{100} \left|\overline{\theta}^{\overline{M},i}_{\star} - \theta_{\star}^L \right|^2,
$$
where  $\theta_{\star}^L$  is a reference value approximating the true maximum likelihood estimate (MLE) obtained through gradient descent. Since the exact likelihood is intractable for these models, we approximate $ \theta_{\star}^L $ by running our algorithm  $2^{15}$  times.
This MLE reference provides a benchmark for evaluating our algorithm's performance. The MSE, therefore, gives insight into the estimator's consistency and convergence by indicating how closely  
$\overline{\theta}^{\overline{M}}_{\star}$  approximates  $\theta_{\star}^L$  over multiple trials.

The number of particles $M$ used in the Conditional Particle Filter (CPF) and Conditional Coupled Particle Filter (CCPF) algorithms scales as $\mathcal{O}(T)$, where  $T$  represents the final time. To optimize the algorithms, we set parameters as follows:  $T_p = 2^p $ and $ N_l = lL$, where  $l$  is the discretization level.
 The probability distributions for level selection are given by  $\mathbb{P}_L(l) = 2^{-l/2\varepsilon} \; \mathbb{I}_{ \{ l_{\mathrm{min}}, .\dots,L\} } (l)$ for some $l_{\mathrm{min}}, L \in \mathbb{N} $, with  $\varepsilon \in (0,1)$  and $\mathbb{I}$ represents the indicator function. The distribution for the parameter $ p $ is $ \mathbb{P}_P(p) \propto 2^p(p+1)\log_2(p+2)$ , balancing growth with a logarithmic term. These carefully chosen parameters and distributions aim to enhance the performance of our algorithms for accurate parameter estimation in dynamic systems.

In Figures (\ref{fig:Results1}) and (\ref{fig:Results2}), we present the relationship between the Mean Squared Error (MSE) and the computational cost of the algorithm. Here, the computational cost is defined as the cumulative sum of the costs of all $\overline{M}$  parallel processes, where each process contributes to the overall effort required by the algorithm. This cumulative cost reflects the total resources consumed when running the algorithm in parallel.
They reveal that as the number of parallel processes  $\overline{M}$ increases, the MSE decreases, following the scaling behavior $\text{MSE} = \mathcal{O} \big( \frac{1}{\overline{M}} \big)$. This indicates that the algorithm achieves greater accuracy with additional parallel processes, effectively reducing the MSE in inverse proportion to  $\overline{M}$. This scaling behavior underscores the efficiency of parallelization in improving accuracy.

\begin{figure}[H]
\centering
\subfloat[]{\includegraphics[width=0.45\textwidth]{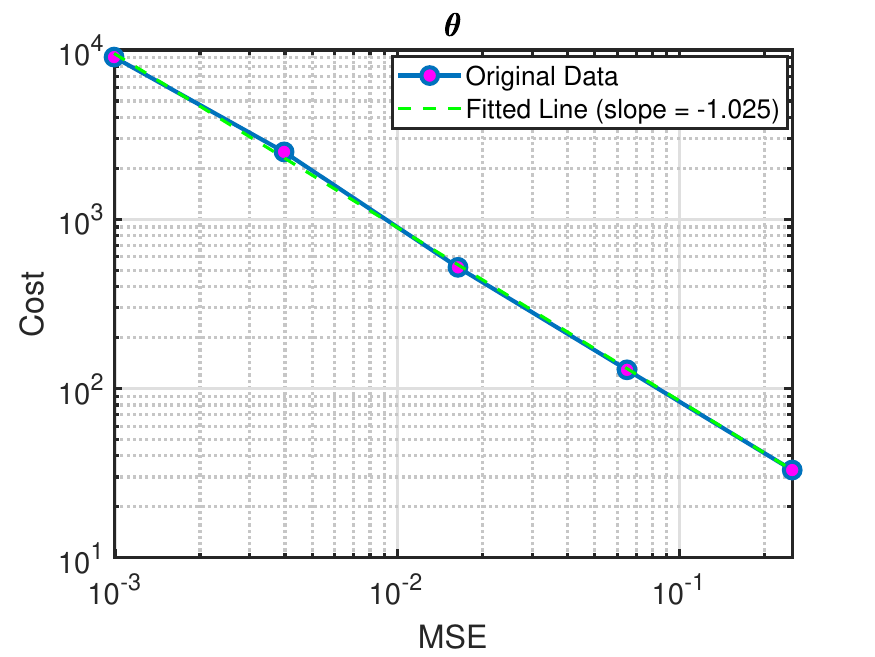}}\qquad
\subfloat[]{\includegraphics[width=0.45\textwidth]{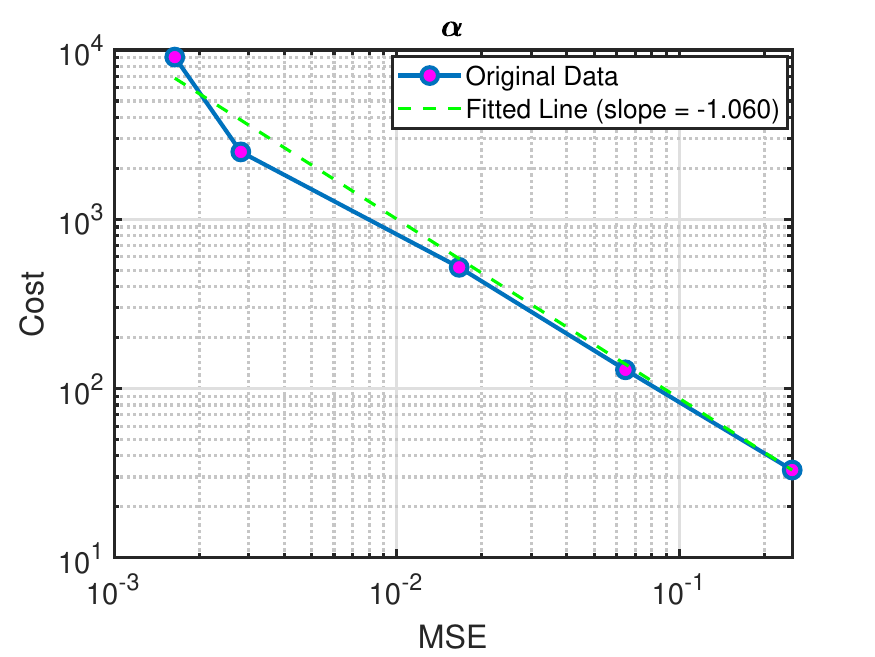}}\qquad
\subfloat[]{\includegraphics[width=0.45\textwidth]{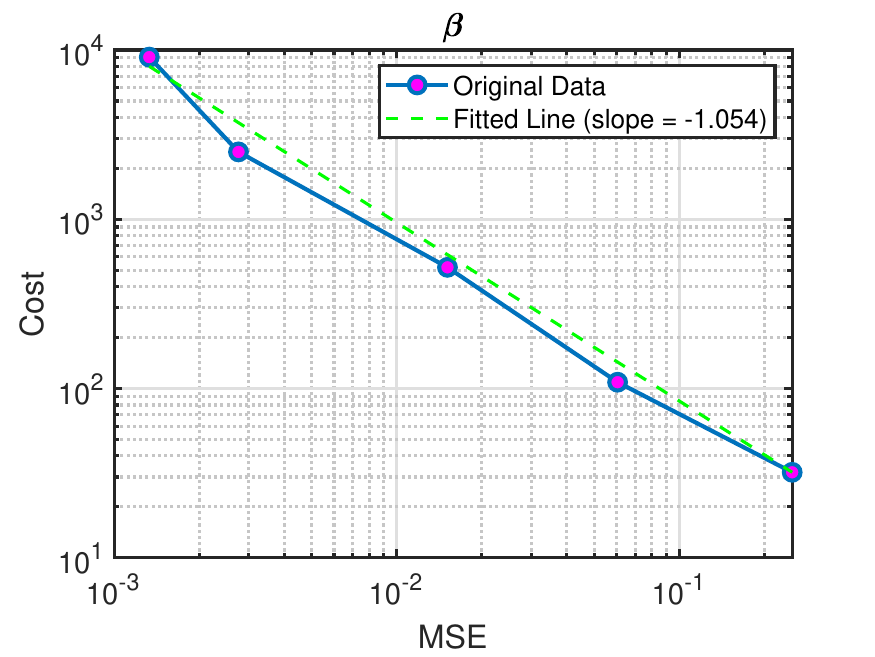}}
\caption{Log-log scale of MSE versus Cost for parameter estimation in the Kuramoto and neural network models, evaluated with $L=4$. (a) Estimation of parameter $\theta$, which governs synchronization dynamics in the Kuramoto model; (b) Estimation of neural network parameter $\alpha$; (c) Estimation of neural network parameter $\beta$.}
\label{fig:Results1}
\end{figure}

\begin{figure}[H]
\centering
\subfloat[]{\includegraphics[width=0.45\textwidth]{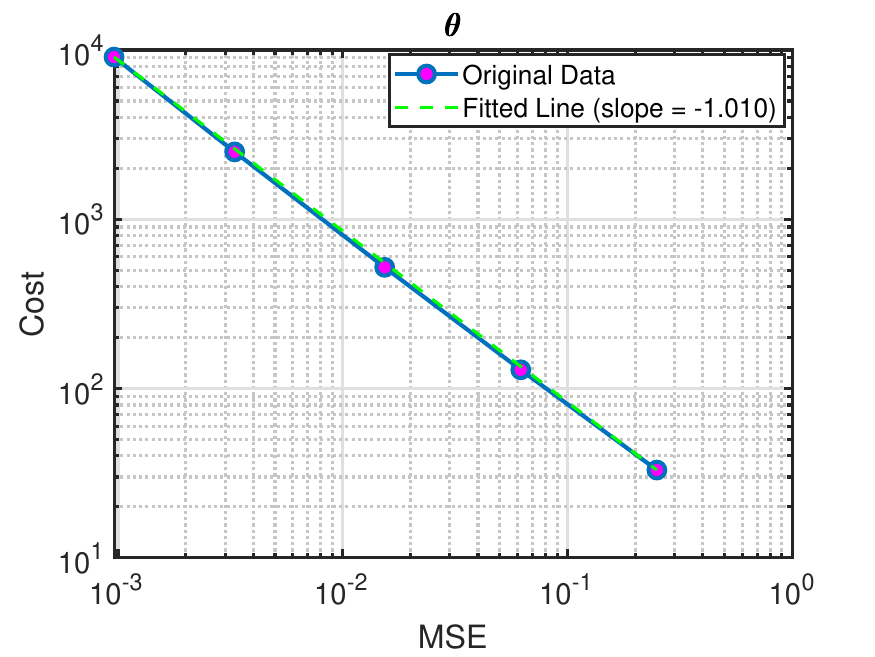}}\qquad
\subfloat[]{\includegraphics[width=0.45\textwidth]{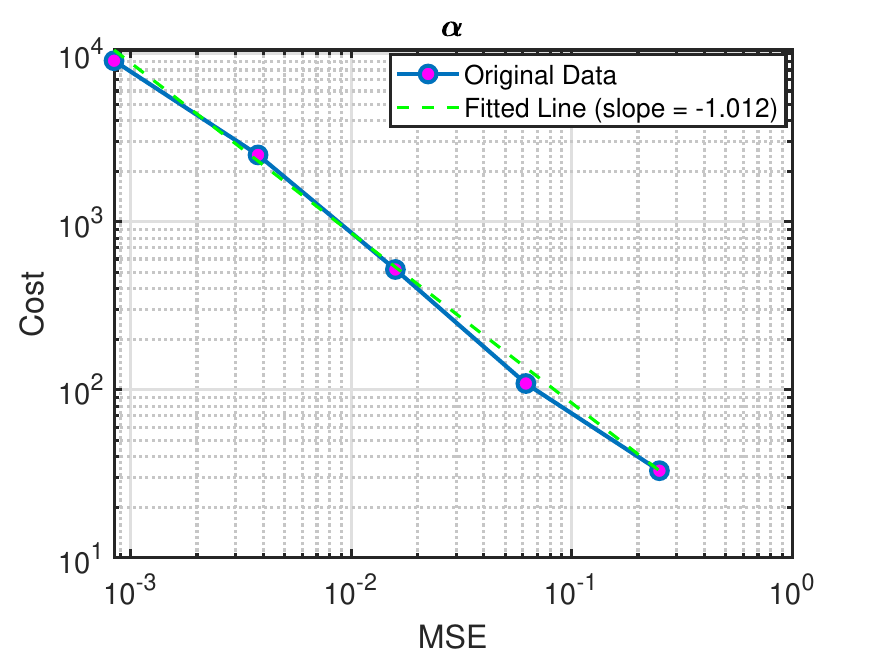}}\qquad
\subfloat[]{\includegraphics[width=0.45\textwidth]{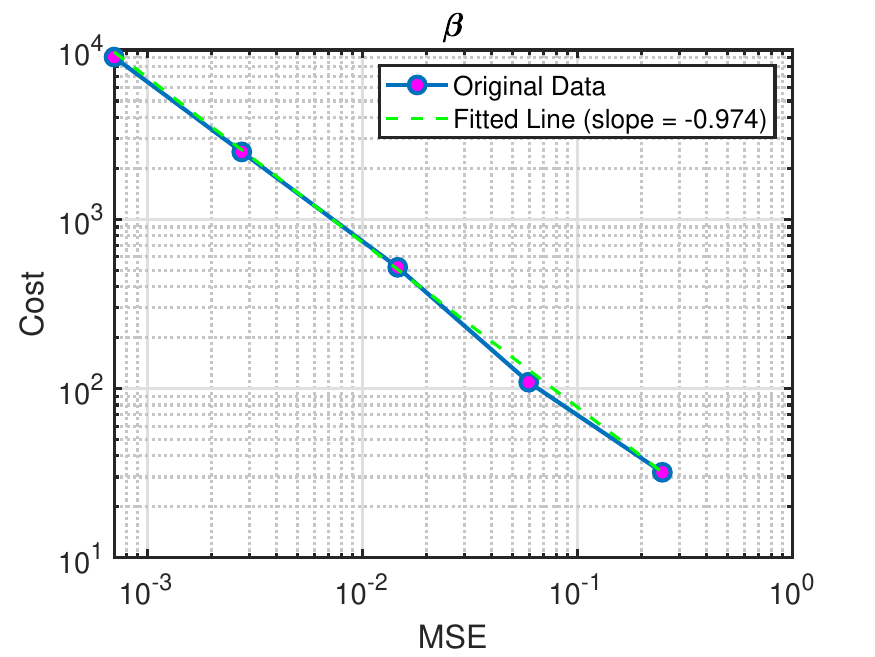}}
\caption{Log-log scale of MSE versus Cost for parameter estimation in the Kuramoto and neural network models, evaluated with $L=8$. (a) Estimation of parameter $\theta$; (b) Estimation of neural network parameter $\alpha$; (c) Estimation of neural network parameter $\beta$.}
\label{fig:Results2}
\end{figure}

\subsubsection*{Acknowledgements}

AJ was supported by CUHK-SZ start-up funding.

\appendix

\section{Theoretical Results}\label{app:theory}

\subsection{Structure}

This appendix is designed to be read in order after the main text has been also read.  The structure of the appendix is as follows.  In Section \ref{app:note} we give some additional notations that are needed.   In Section \ref{app:ass} our assumptions are given with a discussion.  In Section \ref{app:tech} we give several technical results that are needed for the proof of Theorem \ref{theo:main_thm} with the proof of the latter in Section \ref{app:main_theo}.

\subsection{Notations}\label{app:note}

Let $A$ be a symmetric real matrix,  then we denote
by $\textrm{det}(A)$ the determinant of $A$, $\lambda_{\min}(A)$ as the minimum eigenvalue,  $\lambda_{\max}(A)$ as the maximum eigenvalue.  For a $d\times d$ matrix $A$, that is diagonal we write $A=\textrm{diag}_d(a_1,\dots,a_d)$.  The collection of $d\times d$ real positive-definite and symmetric matrices is written $\mathcal{M}^{+}(\mathbb{R}^d)$.
For a given $\mu\in\mathcal{P}(\mathbb{R}^d)$ we write the transition kernel of the first order Euler-Maruyama 
time discretization over a time unit $\Delta$ (see \eqref{eq:sde_disc}) as $Q_{\mu,\theta}^{\Delta}(x,dy)$ and for any measurable functional $\varphi:\mathbb{R}^d\rightarrow\mathbb{R}$ we write $Q_{\mu,\theta}^{\Delta}(\varphi)(x)=\int_{\mathbb{R}^d}
\varphi(y)Q_{\mu,\theta}^{\Delta}(x,dy)$.  In what follows we set $V(x)= \rho x^{\top}x + 1$ for 
$(x,\rho)\in\mathbb{R}^d\times\mathbb{R}^+$  and we set $\mathsf{C}_r=\{x\in\mathbb{R}^d:V(x)\leq r\}$ where
$r>1$ is given.  The indicator function associated to a set $A\in\mathcal{B}(\mathbb{R}^d)$ is written $\mathbb{I}_A(x)$.  In the calculations to follow $C$ is a generic finite and positive constant and if there are any important dependencies (e.g.~on $\theta$)  we will outline them when needed.  
Let $\nu$ be a finite signed measure on $(\mathbb{R}^d,\mathcal{B}(\mathbb{R}^d))$; we denote
by $|\nu|$ the Hahn-Jordan (HJ) decomposition.  For a measurable function $f:\mathbb{R}^d\rightarrow\mathbb{R}$
we set
$$
\|f\|_{e^V} = \sup_{x\in\mathbb{R}^d}\frac{|f(x)|}{e^{V(x)}}.
$$
We also set $|\nu|(f)= \int_{\mathbb{R}^d}f(x)|\nu|(dx)$ where $|\nu|$ is the HJ decomposition of $\nu$ a finite signed measure.
For $x\in\mathbb{R}^d$ we denote the Euclidean norm as $\|x\|$.

\subsection{Model Structure and Assumptions}\label{app:ass}

We now reduce the type of SDEs that we will analyze and in particular,  we shall take $\sigma(x)=\textrm{diag}_d(\widetilde{\sigma}(x),\dots,\widetilde{\sigma}(x))$ for some $\widetilde{\sigma}:\mathbb{R}^d\rightarrow\mathbb{R}^+$.  For the drift coefficient,  we will use the short-hand 
$$
a_{\theta}(x,\mu) = a_{\theta}(x,\overline{\xi}_{\theta}(x,\mu)).
$$
We write the $d\times d$ matrix of first-order $x$ derivatives of $a_{\theta}(x,\mu)$ as $\nabla a_{\theta}(x,\mu)$ that is for $(i,j)\in\{1,\dots,d\}^2$
$$
\nabla a_{\theta}(x,\mu)^{(i,j)}= \frac{\partial a_{\theta}(x,\mu)^{(i)}}{\partial x_j}
$$
where $\nabla a_{\theta}(x,\mu)^{(i,j)}$ is the $(i,j)^{\textrm{th}}$ element of $\nabla a_{\theta}(x,\mu)$
and for a $d-$vector $x$ we write $x^{(k)}$ as the $k^{\textrm{th}}-$element of $x$.
As it will become useful below,  for $\kappa$ a $d-$vector,  we will write 
$\nabla a_{\theta}(\kappa x,\mu)$ to denote the $d\times d$ matrix with elements
for $(i,j)\in\{1,\dots,d\}^2$
$$
\nabla a_{\theta}(\kappa x,\mu)^{(i,j)}= \frac{\partial a_{\theta}(\kappa_i x,\mu)^{(i)}}{\partial x_j}.
$$
We make the following assumptions.

\begin{hypA}\label{ass:1_new}
$\Theta$ is open and bounded.
\end{hypA}

\begin{hypA}\label{ass:2_new}
For each $l\in\{0,\dots,L\}$ the function $\theta\mapsto\widehat{\nabla_{\theta}\log(p_{\theta}^{l,N_l}(y_1,\dots,y_T))}$ is twice continuously differentiable. Moreover there exists a unique root of 
$\widehat{\nabla_{\theta}\log(p_{\theta}^{l,N_l}(y_1,\dots,y_T))}$ which is denoted $\theta_{\star}^L$.
\end{hypA}

\begin{hypA}\label{ass:1}
We have:
\begin{enumerate}
\item{For each $(\theta,\mu)\in\Theta\times\mathcal{P}(\mathbb{R}^d)$  $a_{\theta}(\cdot,\mu)\in \mathcal{C}_b^1(\mathbb{R}^{d+1},\mathbb{R}^{d})$.  $\tilde{\sigma}\in\mathcal{C}_b^1(\mathbb{R}^d,\mathbb{R})$.}
\item{There exists a $C<+\infty$ such that for each $(i,k)\in\{1,\dots,d\}^2$
$$
\max\left\{
\sup_{\theta\in\Theta}\sup_{\mu\in\mathcal{P}(\mathbb{R}^d)}\sup_{x\in\mathbb{R}^d}
\left|\frac{\partial a_{\theta}^{(k)}(x,\mu)}{\partial x_i}\right|,
\sup_{\theta\in\Theta}\sup_{\mu\in\mathcal{P}(\mathbb{R}^d)} |a_{\theta}^{(k)}(0,\mu)|
\right\}
\leq C.
$$
}
\item{We have
\begin{eqnarray*}
\lambda^{\star}(1):=\sup_{\theta\in\Theta}\sup_{\mu\in\mathcal{P}(\mathbb{R}^d)}
\sup_{x\in\mathbb{R}^d}
\sup_{\kappa\in(0,1)^d}
\lambda_{\max}\left(\tfrac{1}{2}\left(\nabla a_{\theta}(\kappa x,\mu)^{\top}+\nabla a_{\theta}(\kappa x,\mu)\right)\right) & < & 0\\
\lambda^{\star}(2):=\sup_{\theta\in\Theta}\sup_{\mu\in\mathcal{P}(\mathbb{R}^d)}
\sup_{x\in\mathbb{R}^d}
\sup_{\kappa\in(0,1)^d}
\lambda_{\max}\left((\nabla a_{\theta}(\kappa x,\mu)^{\top}\nabla a_{\theta}(\kappa x,\mu)\right) & < & 0
\end{eqnarray*}
}
\item{There exists a $0<\underline{C}<\overline{C}<+\infty$ such that for every $x\in\mathbb{R}^d$
$$
\underline{C} \leq  \tilde{\sigma}(x) \leq \overline{C}.
$$
}
\end{enumerate}
\end{hypA}

Set
\begin{equation}\label{eq:lambda_star}
\lambda^{\star}=\Delta^2\lambda^{\star}(1)+2\Delta\lambda^{\star}(2).
\end{equation}

\begin{hypA}\label{ass:2}
There exists a $C<+\infty$ such that for any $(\theta,\theta',k)\in\Theta^2\times\{1,\dots,d\}$
$$
\max\left\{\sup_{\mu\in\mathcal{P}(\mathbb{R}^d)}
\sup_{x\in\mathbb{R}^d}|a_{\theta}(x,\mu)^{(k)}-a_{\theta'}(x,\mu)^{(k)}|,
\sup_{\mu\in\mathcal{P}(\mathbb{R}^d)}
\sup_{x\in\mathbb{R}^d}
|\{a_{\theta}(x,\mu)^{(k)}\}^2-\{a_{\theta'}(x,\mu)^{(k)}\}^2|
\right\} \leq 
$$
$$
C\|\theta-\theta'\|\log\left(\|\theta-\theta'\|\right).
$$
\end{hypA}

\begin{hypA}\label{ass:3}
For each $y\in\mathsf{Y}$ there exists a $C>1$ such that for any $(x,\theta)\in\mathbb{R}^d\times\Theta$
$\tfrac{1}{C}\leq G_{\theta}(x,y), |\nabla_{\theta}G_{\theta}(x,y)| \leq C$.
\end{hypA}

Recall $H_l$ as in \eqref{eq:Hl_def}. For any sequence of probability measures 
$\eta_{0},\dots,\eta_{\Delta_{l}^{-1}T-1}$ where for any $k\in\{0,\dots,\Delta_{l}^{-1}T-1\}$
$\eta_{k}\in\mathcal{P}(\mathbb{R}^d)$,  we will write $H_{l,\eta}(\theta,x_0,\dots,x_T)$ to denote $H_l$ where this particular sequence has been plugged in.
Recall from Section \ref{sec:hl_comp} that we will approximate deriviatives associated to $\overline{\xi}_{\theta}$ via finite differences and we write then this as $\widehat{H}_{l,\eta}(\theta,x_0,\dots,x_T)$.  We will use the notation
$\eta\in\mathcal{P}(\mathbb{R}^d)^{\Delta_{l}^{-1}T}$ to mean that each $\eta_{k}\in\mathcal{P}(\mathbb{R}^d)$
for $k\in\{0,\dots,\Delta_{l}^{-1}T-1\}$.

\begin{hypA}\label{ass:4}
There exists a $(\beta,C)\in(0,1]\times(0,\infty)$ such that for any $(x_{\Delta_l},\dots,x_T)\in\mathsf{E}_l^T$
\begin{eqnarray*}
\sup_{\theta\in\Theta}
\sup_{\eta\in\mathcal{P}(\mathbb{R}^d)^{\Delta_{l}^{-1}T}}
\frac{|\widehat{H}_{l,\eta}(\theta,x_0,\dots,x_T)
|}{\tfrac{1}{T}\sum_{t=1}^T\sum_{k=0}^{\Delta_l^{-1}-1}\exp\{x_{t-1+(k+1)\Delta_l}^{\top}x_{t-1+(k+1)\Delta_l}+1\}} & \leq & C\\
\sup_{(\theta,\theta')\in\Theta^2, \theta\neq\theta'}
\sup_{\eta\in\mathcal{P}(\mathbb{R}^d)^{\Delta_{l}^{-1}T}}
\frac{|\widehat{H}_{l,\eta}(\theta,x_0,\dots,x_T)
-\widehat{H}_{l,\eta}(\theta',x_0,\dots,x_T)
|}{\|\theta-\theta'\|^{\beta}\tfrac{1}{T}\sum_{t=1}^T\sum_{k=0}^{\Delta_l^{-1}-1}\exp\{x_{t-1+(k+1)\Delta_l}^{\top}x_{t-1+(k+1)\Delta_l}+1\}}
& \leq & C.
\end{eqnarray*}
\end{hypA}

\subsubsection{Discussion of Assumptions}

Several of our assumptions deserve some discussion.  (A\ref{ass:1_new}-\ref{ass:2_new}) are fairly typical in the context of parameter estimation of state-space models; see for instance \cite{cappe}.
For (A\ref{ass:1}),  it is known that 1.~is enough to guarantee a unique strong solution to \eqref{eq:sde}; see for instance \cite{ver}.  2.~is needed to control certain expectations that occur in our computations.  In terms of 3.~in the context of regular SDEs it is known (e.g.~\cite{delm}) to be used for stability properties of the SDE and indeed as our analysis is related to the ergodicity of the time-discretization (see \cite{delm1,whiteley} for coverage of these properties) one needs to pass the stability of the SDE to that of the discrete-time approximation. The rest of the assumptions are essentially used to help verify the assumptions in \cite{ub_grad_new} which is the approach we will use and in addition,  the assumptions are in the main discussed there.
 
\subsection{Technical Results}\label{app:tech}

\begin{lem}\label{lem:comp_sq}
Let $X\sim\mathcal{N}_d(\mu,\Delta\Sigma)$, for $(\Delta,\Sigma)\in\mathbb{R}^+\times\mathcal{M}^{+}(\mathbb{R}^d)$.  Let $V(x)=\rho x^{\top}x + 1$,  $(x,\rho)\in\mathbb{R}^d\times\mathbb{R}$ with $0<\rho< \frac{ \lambda_{\min}(\Sigma^{-1}) }{2\Delta}$. Then we have that
$$
\mathbb{E}[\exp\{V(X)\}] = \frac{\textrm{\emph{det}}(\Sigma)^{-1/2}}{\textrm{\emph{det}}(\widetilde{\Sigma})^{1/2}}
\exp\left\{
\frac{1}{2\Delta}\mu^{\top}\left(
\Sigma^{-1}\widetilde{\Sigma}^{-1}\Sigma^{-1}-\Sigma^{-1}
\right)\mu + 1
\right\}
$$
where $\widetilde{\Sigma}=\Sigma^{-1}-\textrm{\emph{diag}}_d(2\Delta\rho,\dots,2\Delta\rho)$.
\end{lem}

\begin{proof}
We note that clearly one can write
$$
\frac{\textrm{det}(\Sigma)^{-1/2}}{(2\pi\Delta)^{d/2}}
\exp\left\{-\frac{1}{2\Delta}(x-\mu)^{\top}\Sigma^{-1}(x-\mu)\right\}
\exp\left\{\rho x^{\top}x + 1\right\} = 
$$
$$
\frac{\textrm{det}(\Sigma)^{-1/2}}{(2\pi\Delta)^{d/2}\textrm{det}(\widetilde{\Sigma})^{1/2}}
\textrm{det}(\widetilde{\Sigma})^{1/2}
\exp\left\{\frac{1}{2\Delta}
\mu^{\top}\left(
\Sigma^{-1}\widetilde{\Sigma}^{-1}\Sigma^{-1}-\Sigma^{-1}
\right)\mu + 1
\right\}
\exp\left\{-\frac{1}{2\Delta}(x-\alpha)^{\top}\widetilde{\Sigma}(x-\alpha)\right\}
$$
where $\alpha=\widetilde{\Sigma}^{-1}\Sigma^{-1}\mu$ and the inverse of $\widetilde{\Sigma}$ exists due to
the condition on $\rho$. The proof is now easily completed by integrating over $x$.
\end{proof}

In the below,  recall \eqref{eq:lambda_star}.
\begin{lem}\label{lem:drift}
Assume (A\ref{ass:1}).  Then for any $\Delta>0$ there exists a $C<+\infty$ such that if
\begin{eqnarray*}
1>\tilde{\phi}& >& \max\{0,1+\xi\} \\
0< \rho & < &\min\left\{
\frac{1}{2\Delta\overline{C}},
\frac{1}{2\Delta\underline{C}}\left(1-\frac{1+\lambda^{\star}}{\tilde{\phi}}\right)
\right\}\\
r& > &1+\left(\frac{C}{1-\tilde{\phi}}\right)^2\\
1> \phi & = & \tilde{\phi} + \frac{C}{\sqrt{r-1}}>0\\
 b & = & C
\end{eqnarray*}
then for any $x\in\mathbb{R}^d$
$$
\sup_{\theta\in\Theta}\sup_{\mu\in\mathcal{P}(\mathbb{R}^d)} Q_{\mu,\theta}^{\Delta}(e^V)(x) \leq \exp\left\{\phi V(x) + b\mathbb{I}_{\mathsf{C}_r}(x)\right\}.
$$
\end{lem}

\begin{proof}
We have by Lemma \ref{lem:comp_sq} that
\begin{equation}\label{eq:main_eq}
Q_{\mu,\theta}^{\Delta}(e^V)(x) = 
\frac{\tilde{\sigma}(x)^{-d/2}}{\left(\tfrac{1}{\tilde{\sigma}(x)}-2\Delta\rho\right)^{d/2}}
\exp\left\{
\frac{\rho}{(1-2\Delta\rho\tilde{\sigma}(x))}
(a_{\theta}(x,\mu)\Delta + x)^{\top}(a_{\theta}(x,\mu)\Delta + x) + 1
\right\}
\end{equation}
which is well-defined due to the condition on $\rho$.
Clearly by (A\ref{ass:1}) 4.~and our assumption on $\rho$ there exists a $C$ such that
\begin{equation}\label{eq:minor_bound}
\frac{\tilde{\sigma}(x)^{-d/2}}{\left(\tfrac{1}{\tilde{\sigma}(x)}-2\Delta\rho\right)^{d/2}} \leq C.
\end{equation}
As a result to complete the proof we will first consider the exponent of the R.H.S.~of \eqref{eq:main_eq}.
We have by using the first mean-value theorem that
\begin{equation}\label{eq:fmvt}
a_{\theta}(x,\mu) = a_{\theta}(0,\mu) + \nabla a_{\theta}(\kappa x,\mu)x 
\end{equation}
where $\kappa$ is a $d-$ vector of values on $(0,1)$.
Then we have that
$$
(a_{\theta}(x,\mu)\Delta + x)^{\top}(a_{\theta}(x,\mu)\Delta + x) = T_1 + T_2
$$
where
\begin{eqnarray*}
T_1 & = & (\{\nabla a_{\theta}(\kappa x,\mu)\Delta+I_d\}x)^{\top}(\{\nabla a_{\theta}(\kappa x,\mu)\Delta+I_d\}x) \\
T_2 & = & 
2\Delta a_{\theta}(0,\mu)^{\top}(\{\nabla a_{\theta}(\kappa x,\mu)\Delta+I_d\}x) 
+ 
\Delta^2 a_{\theta}(0,\mu)^{\top}a_{\theta}(0,\mu)
\end{eqnarray*}
where $I_d$ is the $d\times d$ identity matrix.
We will deal with the two terms $T_1$ and $T_2$ individually. 

For $T_1$ our objective is to show that for every $x\in\mathbb{R}^d$
\begin{equation}\label{eq:t1}
\frac{\rho}{(1-2\Delta\rho\tilde{\sigma}(x))} T_1 \leq \tilde{\phi}V(x)
\end{equation}
where $\tilde{\phi}\in(0,1)$ and does not depend on $\theta,\mu$ but could depend on $\Delta$.
Now if $x=0$ or if $\{\nabla a_{\theta}(\kappa x,\mu)\Delta+I_d\}$ is negative-definite,  verifying 
\eqref{eq:t1} is trivial,  so we suppose that neither is the case.  Then we have 
$$
\frac{\rho}{(1-2\Delta\rho\tilde{\sigma}(x))V(x)} T_1 \leq \frac{x^{\top}
(\{\nabla a_{\theta}(\kappa x,\mu)\Delta+I_d\})^{\top}
(\{\nabla a_{\theta}(\kappa x,\mu)\Delta+I_d\})
x}
{(1-2\Delta\rho\tilde{\sigma}(x))x^{\top}x}.
$$
It is now enough to show that the R.H.S.~of the above inequality is upper-bounded by $\tilde{\phi}$,  which happens if
the matrix
\begin{equation}\label{eq:mat_lya}
\Delta^2 \nabla a_{\theta}(\kappa x,\mu)^{\top}\nabla a_{\theta}(\kappa x,\mu) + 2\Delta \tfrac{1}{2}(\nabla a_{\theta}(\kappa x,\mu)^{\top}+\nabla a_{\theta}(\kappa x,\mu)) + (1-\tilde{\phi}\{1-2\Delta\rho\tilde{\sigma}(x)\}) I_d
\end{equation}
is negative definite.  Since all the matrices $\nabla a_{\theta}(\kappa x,\mu)^{\top}\nabla a_{\theta}(\kappa x,\mu)$,
$\tfrac{1}{2}(\nabla a_{\theta}(\kappa x,\mu)^{\top}+\nabla a_{\theta}(\kappa x,\mu))$ 
and $(1-\tilde{\phi}\{1-2\Delta\rho\tilde{\sigma}(x)\})I_d$
are symmetric so is the sum and all the eigenvalues are real.  The largest eigenvalue is
$$
\Delta^2\lambda_{\max}\left(\nabla a_{\theta}(\kappa x,\mu)^{\top}\nabla a_{\theta}(\kappa x,\mu)\right) + 
2\Delta\lambda_{\max}\left(\tfrac{1}{2}(\nabla a_{\theta}(\kappa x,\mu)^{\top}+\nabla a_{\theta}(\kappa x,\mu))\right) + 
(1-\tilde{\phi}\{1-2\Delta\rho\tilde{\sigma}(x)\}) \leq
$$
\begin{equation}\label{eq:mat_lya1}
\Delta^2\lambda^{\star}(1)+2\Delta\lambda^{\star}(2) + 
(1-\tilde{\phi}\{1-2\Delta\rho\tilde{\sigma}(x)\})
\end{equation}
where $\lambda^{\star}(1)$ and $\lambda^{\star}(2)$ are defined in (A\ref{ass:1}) 3.;
it now suffices to show that \eqref{eq:mat_lya1} is less than zero.
Now choose $1>\tilde{\phi}>\max\{0,1+\Delta^2\lambda^{\star}(1)+2\Delta\lambda^{\star}(2)\}$, 
then the largest eigenvalue of the matrix in \eqref{eq:mat_lya}
is negative when
$$
\rho < \frac{1}{2\Delta\tilde{\sigma}(x)}\left(1-\frac{1+\lambda^{\star}}{\tilde{\phi}}\right) \leq 
\frac{1}{2\Delta\underline{C}}\left(1-\frac{1+\lambda^{\star}}{\tilde{\phi}}\right)
$$
which is possible as the R.H.S.~is positive and we have assumed the condition on $\rho$ in the statement.  Hence we have verified that \eqref{eq:t1} holds and that moreover
$\tilde{\phi}$ does not depend on $\theta$ or $\mu$.

For $T_2$ our objective is to show that for every $x\in\mathbb{R}^d$
\begin{equation}\label{eq:mat_lya2}
\frac{\rho}{(1-2\Delta\rho\tilde{\sigma}(x))} T_2 \leq C\left(\frac{1}{\sqrt{r-1}}\mathbb{I}_{\mathsf{C}_r^c}(x) V(x) + \mathbb{I}_{\mathsf{C}_r}(x)\right)
\end{equation}
where $0<C<+\infty$ and does not depend on $\theta,\mu$ but could depend on $\Delta$ and $r$.
By (A\ref{ass:1}) 2.~we note that for any $x,\theta,\mu$
$$
\frac{\rho}{(1-2\Delta\rho\tilde{\sigma}(x))}
\Delta^2a_{\theta}(0,\mu)^{\top}a_{\theta}(0,\mu) \leq C
$$
where $C$ does not depend on $x,\theta,\mu$ so we need only  consider the first term in the definition of $T_2$.
Then we note that using (A\ref{ass:1}) 2.~and the Cauchy-Schwarz inequality that for any $x,\theta,\mu$
\begin{equation}\label{eq:mat_lya3}
\frac{\rho}{(1-2\Delta\rho\tilde{\sigma}(x))}2\Delta a_{\theta}(0,\mu)^{\top}(\{\nabla a_{\theta}(\kappa x,\mu)\Delta+I_d\}x)  \leq C(x^{\top}x)^{1/2}.
\end{equation}
where $C$ does not depend on $x,\theta,\mu$. Now if $x\in\mathsf{C}_r$ clearly the R.H.S.~of \eqref{eq:mat_lya3} is upper-bounded by a constant $C$ that does not depend on $x,\theta,\mu$ but does depend on $\Delta$ and $r$. If $x\in\mathsf{C}_r^c$ then it is clear that the R.H.S.~of \eqref{eq:mat_lya3} 
when divided by $V(x)$ 
is upper-bounded by a term that is $\mathcal{O}(1/(\sqrt{r-1}))$; this verifies \eqref{eq:mat_lya2}. 

Now combining \eqref{eq:main_eq} with the results that are proved in \eqref{eq:minor_bound}, 
\eqref{eq:t1} and \eqref{eq:mat_lya2} we have shown that
$$
Q_{\mu,\theta}^{\Delta}(e^V)(x) \leq \exp\left\{
\left(
\tilde{\phi} + \frac{C}{\sqrt{r-1}}\mathbb{I}_{\mathsf{C}_r^c}(x)
\right)V(x) + 
C\mathbb{I}_{\mathsf{C}_r}(x)
\right\}.
$$
If one chooses 
$$
r>1+\left(\frac{C}{1-\tilde{\phi}}\right)^2
$$
and then 
$$
\phi = \tilde{\phi} + \frac{C}{\sqrt{r-1}}\quad\quad b=C
$$
we have shown that 
$$
\sup_{\theta\in\Theta}\sup_{\mu\in\mathcal{P}(\mathbb{R}^d)} Q_{\mu,\theta}^{\Delta}(e^V)(x) \leq \exp\left\{\phi V(x) + b\mathbb{I}_{\mathsf{C}_r}(x)\right\}
$$
as was to be proved.
\end{proof}

\begin{lem}
Assume (A\ref{ass:1}).  Then for any $\Delta>0$ and $r>1$ there exists $(\nu,\epsilon)\in\mathcal{P}(\mathbb{R}^d)\times(0,1)$ such that for any $x\in\mathsf{C}_r$:
$$
\inf_{\theta\in\Theta}\inf_{\mu\in\mathcal{P}(\mathbb{R}^d)} Q_{\mu,\theta}^{\Delta}(x,dy) \geq \epsilon \nu(dy).
$$
\end{lem}

\begin{proof}
We have that for any $x,\theta,\mu$.
\begin{equation}\label{eq:minor1}
Q_{\mu,\theta}^{\Delta}(x,dy) \geq C\exp\left\{-C(y-\{a_{\theta}(x,\mu)\Delta + x\})^{\top}
(y-\{a_{\theta}(x,\mu)\Delta + x\})
\right\}dy
\end{equation}
where $C$ is a positive and finite constant that depends on $\Delta$ but not $x,\theta,\mu$ and $dy$ is the $d-$dimensional Lebesgue measure.  Now, as in the proof of Lemma \ref{lem:drift}, applying 
the first mean value theorem the exponent of the R.H.S.~of \eqref{eq:minor1} is equal to
$$
-C\left(y^{\top}y + T_1 + T_2\right)
$$
where
\begin{eqnarray*}
T_1 & = & -2\left(\{\nabla a_{\theta}(\kappa x,\mu)\Delta+I_d\}x + \Delta a_{\theta}(0,\mu)\right)^{\top}y \\
T_2 & = & \left(\{\nabla a_{\theta}(\kappa x,\mu)\Delta+I_d\}x + \Delta a_{\theta}(0,\mu)\right)^{\top}\left(\{\nabla a_{\theta}(\kappa x,\mu)\Delta+I_d\}x + \Delta a_{\theta}(0,\mu) \right).
\end{eqnarray*} 
We will now seek to lower-bound $e^{-CT_1}$ and $e^{-CT_2}$ individually,  when $x\in\mathsf{C}_r$.

In the case of $T_1$,  by (A\ref{ass:1}) 2.~,  we know that for each $k\in\{1,\dots,d\}$
\begin{equation}\label{eq:minor2}
\sup_{\theta\in\Theta}\sup_{\mu\in\mathcal{P}(\mathbb{R}^d)}|\Delta a_{\theta}(0,\mu) ^{(k)}| \leq C
\end{equation}
where $C$ depends on $\Delta$.  Moreover,  by (A\ref{ass:1}) 2.~and Cauchy-Schwarz, for each $k\in\{1,\dots,d\}$
\begin{eqnarray}
\sup_{\theta\in\Theta}\sup_{\mu\in\mathcal{P}(\mathbb{R}^d)}\sup_{x\in\mathbb{R}^d}|
(\{\nabla a_{\theta}(\kappa x,\mu)\Delta+I_d\}x)^{(k)}| & \leq &  C (x^{\top}x)^{1/2}\nonumber\\
&\leq & C\left(\frac{r-1}{\rho}\right)^{1/2}\label{eq:minor3}
\end{eqnarray}
where the second follows as $x\in\mathsf{C}_r$ and $\rho$ is as in $V(x)$.  Due to the bounds in 
\eqref{eq:minor2}-\eqref{eq:minor3} it is straightforward to show that there exists a constant $0<C<+\infty$ independent of $x,\theta,\mu$ that
$$
e^{-CT_1} \geq  \exp\left\{C\varpi(y)^{\top}y\right\}\mathbb{I}_{\mathsf{C}_r}(y)
$$
where $\varpi(y) = (-\textrm{sgn}(y_1),\dots,-\textrm{sgn}(y_d)^{\top}$.

For the case $T_2$ the bounds \eqref{eq:minor2}-\eqref{eq:minor3} suffice to establish that 
there exists a constant $0<C<+\infty$ independent of $x,\theta,\mu$ such that
for $x\in\mathsf{C}_r$ 
$$
e^{-CT_2} \geq C.
$$
We have therefore shown that for $x\in\mathsf{C}_r$
$$
\inf_{\theta\in\Theta}\inf_{\mu\in\mathcal{P}(\mathbb{R}^d)} Q_{\mu,\theta}^{\Delta}(x,dy) \geq
C \exp\left\{-Cy^{\top}y\right\}\exp\left\{C\varpi(y)^{\top}y\right\}\mathbb{I}_{\mathsf{C}_r}(y)dy
$$
and from here the proof can be easily completed.
\end{proof}

\begin{lem}
Assume (A\ref{ass:1_new}, A\ref{ass:1}-\ref{ass:2}). Then for any $\Delta>0$ there exists a $(\zeta,C)\in(\mathbb{R}^+)^2$
such that for any $(\theta,\theta')\in\Theta^2$
$$
\sup_{\mu\in\mathcal{P}(\mathbb{R}^d)}\||Q_{\mu,\theta}^{\Delta}-Q_{\mu,\theta'}^{\Delta}|(e^V)\|_{e^V} \leq C\|\theta-\theta'\|^{\zeta}.
$$
\end{lem}

\begin{proof}
Define $\mathsf{P}$ and $\mathsf{P}^c$ as the (essentially) unique decomposition of $\mathbb{R}^d$
such that for any $\mathsf{A}\in\mathcal{B}(\mathbb{R}^d)$ with $A\subseteq\mathsf{P}$ we have
$Q_{\mu,\theta}^{\Delta}(\mathbb{I}_\mathsf{A})(x)-Q_{\mu,\theta'}^{\Delta}(\mathbb{I}_\mathsf{A})(x)\geq 0$,
and for $A\subseteq\mathsf{P}^c$ 
$Q_{\mu,\theta}^{\Delta}(\mathbb{I}_\mathsf{A})(x)-Q_{\mu,\theta'}^{\Delta}(\mathbb{I}_\mathsf{A})(x)< 0$
and note that we can allow $\mathsf{P}$ to depend on $x,\mu,\theta,\theta',\Delta$.
Then define
\begin{eqnarray*}
\Psi_{\mu,\theta}^{\Delta}(x) & := &
\frac{\tilde{\sigma}(x)^{-d/2}}{\left(\tfrac{1}{\tilde{\sigma}(x)}-2\Delta\rho\right)^{d/2}}
\exp\left\{\frac{\rho\tilde{\sigma}(x)}{1-2\Delta\rho\tilde{\sigma}(x)}
(a_{\theta}(x,\mu)\Delta + x)^{\top}(a_{\theta}(x,\mu)\Delta + x)
+ 1
\right\}\\
M_{\mu,\theta}^{\Delta}(x,y) & := &
\frac{\left(\tfrac{1}{\tilde{\sigma}(x)}-2\Delta\rho\right)^{d/2}}{(2\pi\Delta)^{d/2}}
\exp\Bigg\{-\frac{1-2\Delta\rho\tilde{\sigma}(x)}{2\Delta\tilde{\sigma}(x)}
\left(y-\frac{1}{1-2\Delta\rho\tilde{\sigma}(x)}(a_{\theta}(x,\mu)\Delta + x)\right)^{\top}\times \\ & &
\left(y-\frac{1}{1-2\Delta\rho\tilde{\sigma}(x)}(a_{\theta}(x,\mu)\Delta + x)\right)
\Bigg\}
\end{eqnarray*}
and we remark that 
\begin{equation}\label{eq:psi_v_link}
\Psi_{\mu,\theta}^{\Delta}(x)=Q_{\mu,\theta}^{\Delta}(e^V)(x).
\end{equation}
Then we have that
$$
|Q_{\mu,\theta}^{\Delta}-Q_{\mu,\theta'}^{\Delta}|(e^V)(x) = T_1 + T_2
$$
where
\begin{eqnarray*}
T_1 & = & 2\Psi_{\mu,\theta}^{\Delta}(x)\int_{\mathbb{R}^d}
\left\{M_{\mu,\theta}^{\Delta}(x,y)-M_{\mu,\theta'}^{\Delta}(x,y)\right\}\mathbb{I}_{\mathsf{P}}(y)dy \\
T_2 & = & \left\{\Psi_{\mu,\theta}^{\Delta}(x)-\Psi_{\mu,\theta'}^{\Delta}(x)\right\}\int_{\mathbb{R}^d}
M_{\mu,\theta'}^{\Delta}(x,y)\left\{\mathbb{I}_{\mathsf{P}}(y)-\mathbb{I}_{\mathsf{P}^c}(y)\right\}dy.
\end{eqnarray*}
The proof will now seek to appropriately control the terms $T_1$ and $T_2$.

For $T_1$ we have that the integrand is
$$
\left\{M_{\mu,\theta}^{\Delta}(x,y)-M_{\mu,\theta'}^{\Delta}(x,y)\right\}\mathbb{I}_{\mathsf{P}}(y) =
\left\{\frac{M_{\mu,\theta}^{\Delta}(x,y)}{M_{\mu,\theta'}^{\Delta}(x,y)}-1\right\}\mathbb{I}_{\mathsf{P}}(y)
M_{\mu,\theta'}^{\Delta}(x,y)
$$
and we shall focus on the ratio of the R.H.S.~of the above equation. Clearly,  we have
\begin{eqnarray*}
\frac{M_{\mu,\theta}^{\Delta}(x,y)}{M_{\mu,\theta'}^{\Delta}(x,y)} & = & 
\exp\Bigg\{-\frac{1-2\Delta\rho\tilde{\sigma}(x)}{2\Delta\tilde{\sigma}(x)}\Bigg[
\left(y-\frac{1}{1-2\Delta\rho\tilde{\sigma}(x)}(a_{\theta}(x,\mu)\Delta + x)\right)^{\top}\times \\ & &
\left(y-\frac{1}{1-2\Delta\rho\tilde{\sigma}(x)}(a_{\theta}(x,\mu)\Delta + x)\right) - 
\left(y-\frac{1}{1-2\Delta\rho\tilde{\sigma}(x)}(a_{\theta'}(x,\mu)\Delta + x)\right)^{\top}
-
\\ & & 
\left(y-\frac{1}{1-2\Delta\rho\tilde{\sigma}(x)}(a_{\theta'}(x,\mu)\Delta + x)\right)^{\top}\Bigg]
\Bigg\}.
\end{eqnarray*}
The exponent can written as $T_3 + T_4$ where
\begin{eqnarray*}
T_3 & = & -\frac{1}{\Delta\tilde{\sigma}(x)}
\left\{(a_{\theta'}(x,\mu)\Delta + x)^{\top}-
(a_{\theta}(x,\mu)\Delta + x)^{\top}\right\}y\\
T_4 & = & -
\frac{1}{2\Delta\tilde{\sigma}(x)\left(1-2\Delta\rho\tilde{\sigma}(x)\right)}
\left[
(a_{\theta}(x,\mu)\Delta + x)^{\top}
(a_{\theta}(x,\mu)\Delta + x)
-
(a_{\theta'}(x,\mu)\Delta + x)^{\top}
(a_{\theta'}(x,\mu)\Delta + x)
\right].
\end{eqnarray*}
We will now bound $e^{T_3}$ and $e^{T_4}$ in turn. For $T_3$ we have
$$
T_3 = -\frac{1}{\tilde{\sigma}(x)}
\left\{
a_{\theta'}(x,\mu)^{\top} - 
a_{\theta}(x,\mu)^{\top}
\right\}y
$$
and applying (A\ref{ass:1}) 4.~and (A\ref{ass:2}) it follows that 
$$
T_3 \leq C\|\theta-\theta'\|\log\left(\|\theta-\theta'\|\right)|y|.
$$
For $T_4$ we have
$$
T_4 = -\frac{1}{2\tilde{\sigma}(x)\left(1-2\Delta\rho\tilde{\sigma}(x)\right)}
\left(
2\left\{
a_{\theta}(x,\mu)^{\top} - 
a_{\theta'}(x,\mu)^{\top}
\right\}y
+ \Delta \left\{
a_{\theta}(x,\mu)^{\top}a_{\theta}(x,\mu) - 
a_{\theta'}(x,\mu)^{\top}a_{\theta'}(x,\mu)
\right\}
\right).
$$
Again, one can apply (A\ref{ass:1}) 4.~and (A\ref{ass:2}) to obtain that
$$
T_4 \leq C\|\theta-\theta'\|\log\left(\|\theta-\theta'\|\right)(1+|y|).
$$
Putting together the bounds on $T_3$ and $T_4$ we can deduce that the integrand of $T_1$ is upper-bounded by
$$
\|\theta-\theta'\|^C\exp\left\{C|y|\right\}M_{\mu,\theta'}^{\Delta}(x,y)
$$
as $\Theta$ is bounded.  Therefore,  by computing exactly the (conditional) expectation of $\exp\left\{C|y|\right\}$
w.r.t.~the Markov kernel $M_{\mu,\theta'}^{\Delta}(x,y)dy$ we obtain that
$$
T_1\leq C \Psi_{\mu,\theta}(x)\|\theta-\theta'\|^C
\exp\left\{
C\left(|a_{\theta'}(x,\mu)\Delta + x| + 1\right)
\right\}.
$$
where $C$ could depend on $\Delta$ but it does not depend on $x,\theta,\theta',\mu$.
Recalling \eqref{eq:psi_v_link},  by Lemma \ref{lem:drift} there exists a $\rho,\phi,C$, with $0<\phi<1$ that depend on $\Delta$ but not on
$x,\theta,\theta',\mu$ such that
$$
T_1\leq Ce^{\phi V(x) + C}\|\theta-\theta'\|^C
\exp\left\{
C\left(|a_{\theta'}(x,\mu)\Delta + x| + 1\right)
\right\}. 
$$
Then dividing by $e^{V(x)}$ we have
$$
\frac{T_1}{e^{V(x)}}\leq C e^{(\phi-1) V(x)}\|\theta-\theta'\|^C
\exp\left\{
C\left(|a_{\theta'}(x,\mu)\Delta + x| + 1\right)
\right\}.
$$
by using the first mean value theorem (see \eqref{eq:fmvt}) it then follows that for any $x,\theta,\theta',\mu$
$$
\frac{T_1}{e^{V(x)}} \leq C\|\theta-\theta'\|^C
$$
where, again, $C$ does not depend on $x,\theta,\theta',\mu$.  The proof for $T_2$ is very similar,  that is, using almost the same approach one can verify that 
$$
\frac{T_2}{e^{V(x)}} \leq C\|\theta-\theta'\|^C.
$$
As many of the calcuations are repeated, the entire argument is omitted for brevity.
This concludes the proof. 
\end{proof}

\subsection{Proof of Theorem \ref{theo:main_thm}}\label{app:main_theo}

\begin{proof}
The proof is similar to that of \cite[Theorem 4.1]{ub_grad_new}.  We have either assumed or proved all the assumptions (A1-9) used to prove \cite[Theorem 4.1]{ub_grad_new} in the previous sections of this appendix.  Now we note that although the algorithms that we use are different,  the fact that our results (or assumptions) in this paper are uniform in the input probability measure(s) allows us to conclude with an almost identical proof.  
\end{proof}

\end{document}